\documentclass{article}

\PassOptionsToPackage{numbers, compress}{natbib}
\usepackage[preprint]{neurips_2019}
\usepackage{amssymb,amsmath,amsthm}
\usepackage{braket}
\usepackage{mathtools}
\usepackage{algorithm,algorithmic}

\usepackage{color}
\newtheorem{theorem}{Theorem}
\newtheorem{corollary}[theorem]{Corollary}
\newtheorem{definition}[theorem]{Definition}
\newtheorem{lemma}[theorem]{Lemma}
\newtheorem{claim}[theorem]{Claim}

\newcommand{\comment}[1]{}

\newcommand{\abs}[1]{\left\lvert#1\right\rvert}
\newcommand{\norm}[1]{\left\lVert#1\right\rVert}
\newcommand{\idty}[1]{\mathbb{1}}
\newcommand{\ovsqrt}[1]{\frac{1}{\sqrt{2}}}

\newcommand{\gradtheta}[1]{\frac{\partial #1}{\partial \theta}}
\newcommand{\tr}[1]{\mathrm{Tr} \left[#1\right]}
\newcommand{\pr}{{\rm Pr}}
\newcommand{\partr}[2]{\mathrm{Tr}_{#1} \left[#2\right]}
\newcommand{\trh}[1]{\mathrm{Tr}_h \left[#1\right]}

\usepackage[utf8]{inputenc} % allow utf-8 input
\usepackage[T1]{fontenc}    % use 8-bit T1 fonts
\usepackage{hyperref}       % hyperlinks
\usepackage{url}            % simple URL typesetting
\usepackage{booktabs}       % professional-quality tables
\usepackage{amsfonts}       % blackboard math symbols
\usepackage{nicefrac}       % compact symbols for 1/2, etc.
\usepackage{microtype}      % microtypography

\title{Generative training of quantum Boltzmann machines with hidden units}

\author{%
    Nathan Wiebe\\
    Microsoft Research,\\
    \texttt{nawiebe@microsoft.com}
    \And
    Leonard Wossnig\thanks{corresponding author}\\
    University College London, Rahko,\\
    \texttt{leonard.wossnig@cs.ucl.ac.uk} \\
}

\begin{document}

\maketitle

\begin{abstract}
In this article we provide a method for fully quantum generative training of quantum Boltzmann machines with both visible and hidden units while using quantum relative entropy as an objective.
This is significant because prior methods were not able to do so due to mathematical challenges posed by the gradient evaluation.
We present two novel methods for solving this problem.
The first proposal addresses it, for a class of restricted quantum Boltzmann machines with mutually commuting Hamiltonians on the hidden units, by using a variational upper bound on the quantum relative entropy.
The second one uses high-order divided difference methods and linear-combinations of unitaries to approximate the exact gradient of the relative entropy for a generic quantum Boltzmann machine.
Both methods are efficient under the assumption that Gibbs state preparation is efficient and that the Hamiltonian are given by a sparse row-computable matrix.
\end{abstract}

\section{Introduction}
The objective of quantum machine learning is to understand the ability of agents to learn in quantum mechanical settings~\cite{biamonte2017quantum,ciliberto2018quantum,servedio2004equivalences,arunachalam2018optimal,lloyd2013quantum,wiebe2019language}.
One aspect that obstructs this goal is the fact that quantum state vectors lie in an exponentially large vector space~\cite{mcclean2018barren,schuld2018circuit}.
Owing to the size of these vectors, generative models play a central role in quantum machine learning as they can be used to give concise descriptions of these complicated quantum states~\cite{kieferova2016tomography,schuld2019quantum,romero2017quantum,benedetti2019adversarial}.
Various approaches have been put forward to solve this problem~\cite{kieferova2016tomography,romero2017quantum,kappen2018learning,amin2018quantum,crawford2016reinforcement},
but to date all proposed solutions suffer from problems such as requiring classical input data, yielding exponentially small gradients, or an inability to learn with hidden units.
Here we present a new approach to training quantum Boltzmann machines~\cite{kieferova2016tomography,benedetti2017quantum,amin2018quantum} that resolves all these problems, and therefore addresses a major open problem in quantum machine learning.

%Within the last several years quantum machine learning has emerged as one of the major motivations for building a quantum computer.
%Central to this is the concept that quantum computing may be able to provide models for data that can be trained faster, are more private or provide better models for data than classical models.  Perhaps the most evocative intuition behind why quantum machine learning may provide better models for data stems from the fact that quantum computers are widely believed to be able to model probability distributions that are classically hard to sample from.  In effect, this ability gives a broader or richer family of distributions that can be described using a polynomial sized classical circuit.

%The clearest cases where quantum machine learning can provide such advantages occur for quantum data.
%Classical data is usually fed to a quantum algorithm in the form of a training, or test, set of vectors.
Just as for the classical case, i.e., generative training on classical computers, the main goal in quantum generative training is to build a model that allows to sample from a distribution over quantum state vectors that mimics some training distribution.
The natural analog of such a quantum training set would be a density operator which we denote $\rho$, which is a positive semi-definite trace-$1$ Hermitian matrix that (roughly speaking) describes a probability distribution over quantum state vectors.
The goal in quantum generative training is to find a process $V$ by sampling from $\rho$ such that $V: \ket{0} \mapsto \sigma$ (where $\ket{\cdot}$ is a quantum state vector, i.e., just a unit norm vector) such that some chosen distance measure, e.g., $\|\rho -\sigma\|_1$, is small.
Such a task corresponds, in quantum information language, to partial tomography~\cite{kieferova2016tomography} or approximate cloning~\cite{chefles1999strategies}.
%Alternatively supervised learning tasks are possible where the task is not to replicate the distribution but rather to replicate the conditional probability distributions over a label subspace.  This approach is frequently taken in QAOA-based quantum neural networks.

Boltzmann machines, a physics inspired class of neural network~\cite{aarts1988simulated,salakhutdinov2007restricted,tieleman2008training,le2008representational,salakhutdinov2009deep,salakhutdinov2010efficient,lee2009convolutional,hinton2012practical},
have found numerous applications over the last decade~\cite{lee2009unsupervised,lee2009convolutional,mohamed2011acoustic,srivastava2012multimodal}.
have recently gained popularity as a method to model complex quantum systems~\cite{carleo2017solving,torlai2016learning,nomura2017restricted}.
One of the features of Boltzmann machines is that they directly resemble the quantum physics that is inherent in a quantum computer.
In particular, a Boltzmann machine provides an energy for every configuration of a system and then generates samples from a distribution with probabilities that vary exponentially with this energy.
Indeed the distribution it represents is just the canonical ensemble in statistical physics.
The explicit model in this case is
\begin{equation}
  \label{eq:hidden_unit_QBM}
\sigma_v(H) = {\rm Tr}_h \left(\frac{e^{-H}}{Z} \right)= \frac{{\rm Tr}_h~ e^{-H}}{{\rm Tr}~ e^{-H}},
\end{equation}
where ${\rm Tr}_h(\cdot)$ is the partial trace over an auxillary sub-system known as the Hidden subsystem which serves to build correlations between the ``visible'' system.
For classical Boltzmann machines $H$ is an energy function, i.e., a diagonal matrix, but for quantum Boltzmann machines it is a Hermitian matrix known as the Hamiltonian, which has off-diagonal entries.
Notably, if we want to simulate a quantum system then the dimension of $H$ grows exponentially with the number of units $n$ in the sytem, i.e., $H \in \mathbb{C}^{2^n \times 2^n}$.
Thus the goal of generative quantum Boltzmann training is to find a set of parameters $\theta$ (also called weights) that specify a Hamiltonian such that $H = {\rm argmin}_H \left({\rm dist}(\rho,\sigma_v(H))\right)$ for an appropriate distance, or divergence, function.
As an example, a quantum analogue of an all-visible Boltzmann machine with $n_v$ units could take the form
\begin{equation}
    H(\theta) = \sum_{n=1}^{n_v} \theta_{2n-1} {\sigma_x}^{(n)} + \theta_{2n} \sigma_z^{(n)} + \sum_{n>n'} \theta_{(n,n')} \sigma_z^{(n)} \sigma_z^{(n')}.
\end{equation}
Here $\sigma_z^{(n)}$ and $\sigma_x^{(n)}$ are Pauli matrices acting on qubit (unit) $n$.
%We provide a formal definition of a quantum Boltzmann machine below.
%\begin{definition}
%A quantum Boltzmann machine to be a quantum mechanical system that acts on a tensor product of Hilbert spaces $\mathcal{H}_v\otimes \mathcal{H}_h \in \mathbb{C}^{2^n}$ that correspond to the visible and hidden subsystems of the Boltzmann machine.  It further has a Hamiltonian of the form $H \in \mathbb{C}^{2^n \times 2^n}$ such that $\norm{H - \text{diag}(H)} > 0$.  The quantum Boltzmann machine takes these parameters and then outputs a state of the form ${\rm Tr}_h \left(\frac{e^{-H}}{{\rm Tr}(e^{-H})} \right)$.
%\end{definition}

For generative training on classical data the natural divergence to use between the input and output distributions is the KL divergence.
In the case where the input is a quantum state the natural notion of distance changes to the quantum relative entropy:
\begin{equation}
  \label{eq:quant_rel_ent}
S(\rho | \sigma_v) = {\rm Tr}\left( \rho \log \rho \right) - {\rm Tr}\left( \rho \log \sigma_v \right),
\end{equation}
which reduces to the KL divergence if $\rho$ and $\sigma_v$ are diagonal matrices and is zero if and only if $\rho = \sigma_v$.

While the relative entropy is generally difficult to compute, the gradient of the relative entropy is straight forward to compute for Boltzmann machines with all visible units.
This is straight forward because in such cases $\sigma_v = e^{-H}/Z$ and the fact that $\log(e^{-H}/Z) = -H -\log(Z)$ allows the matrix derivatives to be easily computed.
However, no methods are known for the generative training of Boltzmann machines for the quantum relative entropy loss function, if hidden units are present.
This is because the partial trace in $\log({\rm Tr}_h e^{-H} /Z)$ prevents us from simplifying the logarithm term when computing the gradient.

\paragraph{Our Contribution:}
In this work we provide practical methods for training generic quantum Boltzmann machines that have both hidden as well as visible units.
We provide two new approaches for achieving this.
The first, and more efficient of the two, works by assuming a special form for the Hamiltonian that allows us to find a variational upper bound on the quantum relative entropy.
Using this upper bound, the derivatives are easy to compute.
The second, and more general method uses recent techniques from quantum simulation to approximate the exact expression for the gradient using Fourier series approximations and high-order divided difference formulas in place of the analytic derivative.
Both methods are efficient, given that Gibbs state preparation is efficient which we expect to hold in most practical cases of Boltzmann training, although it is worth noting that efficient Gibbs state preparation in general would imply $\mathrm{QMA}\subseteq \mathrm{BQP}$ which is unlikely to hold.

\section{Training Boltzmann Machines}
We now present methods for training quantum Boltzmann machines.\\
The quantum relative entropy cost function for a quantum Boltzmann machine (QBM) with hidden units is given by
\begin{equation}
  \label{eq:obj_QBM_hidden}
 \mathcal{O}_{\rho}(H) = S \left(\rho \Big| \partr{h}{ e^{-H}/\tr{e^{-H}}} \right),
\end{equation}
where $S(\rho | \sigma_v)$ is the quantum relative entropy as defined in eq.~\ref{eq:quant_rel_ent}.
Note that we can add a regularization term to the quantum relative entropy to penalize unnecessary quantum correlations in the model~\cite{kieferova2016tomography}.
In this work, we generally aim to train the QBM with a gradient-based method.
For this we are require to evaluate the gradient of the cost function, and hence the gradient of the quantum relative entropy.

In the case of an all-visible Boltzmann machine (which corresponds to $\text{dim}(\mathcal{H}_h)=1$) a closed form expression for the gradient of the quantum relative entropy is known:
\begin{equation}
 \frac{\partial \mathcal{O}_{\rho}(H)}{\partial \theta} = -\tr{ \frac{\partial}{\partial \theta} \rho \log \sigma},
\end{equation}
which can be simplified using $\log(\exp(-H)) = -H$ and Duhamels formula to obtain the following equation for the gradient, denoting $\partial_{\theta} := \partial/\partial \theta$,
\begin{equation}
  \tr{\rho \partial_{\theta} H} - \tr{e^{-H}\partial_{\theta} H}/\tr{e^{-H}}.
\end{equation}
However, the above gradient formula is not generally valid, and indeed does not hold if we include hidden units.
Allowing for these, we need to additionally trace out the subsystem which results in the majorised distribution from eq.~\ref{eq:hidden_unit_QBM}.
This also changes the cost function which takes then the form described in eq.~\ref{eq:obj_QBM_hidden}.
Note that $H=H(\theta)$ is depending on the variables we will alter during the training process, while $\rho$ is the target density matrix, i.e., the input data.
Therefore, if we want to estimate the gradient of the above, omitting the $\tr{\rho \log \rho}$ since it is a constant, we obtain
\begin{equation}
  \label{eq:grad_QBM_hidden_units}
 \frac{\partial \mathcal{O}_{\rho}(H)}{\partial \theta} = -\tr{ \frac{\partial}{\partial \theta} \rho \log \sigma_v},
\end{equation}
for the gradient of the objective function.

In the following, we discuss two different approaches for evaluating the gradient in eq.~\ref{eq:grad_QBM_hidden_units}.
While the first is less general, it gives an easy implementable algorithm and strong bounds based on optimizing a variational bound.
The second approach is on the other hand applicable to any problem instance, and hence a general purpose gradient optimisation algorithm for relative entropy training.
The no-free-lunch theorem suggests that no (good) bounds can be obtained without assumptions on the problem instance, and indeed, the general algorithm exhibits potentially exponentially worse complexity.
However, for many practical applications we assume that this will not be the case, and in particularly, the result presented gives a generally applicable algorithm for training quantum Boltzmann machines on a quantum device, which is to the best of our knowledge the first known result of this kind.

\subsection{Variational training for restricted Hamiltonians}

Our first approach is based on optimizing a variational bound of the objective function, i.e., the quantum relative entropy, in a restricted - but still practical setting.
This approach will give us a fast and easy to implement quantum algorithm which, however, is less general applicable due to the requirement of certain input assumptions.
These assumptions are important, as several instances of scalar calculus fail when we transit to matrix functional analysis, which particularly applies to the gradient of the quantum relative entropy, and we require these in order to obtain a feasible analytical solution.

We express the Hamiltonian in this case as
\begin{equation}
H = H_v + H_h + H_{\rm int},
\end{equation}
i.e., a decomposition of the Hamiltonian into a part acting on the visible layers, the hidden layers and a third interaction Hamiltonian that creates correlations between the two.
In particular, we further assume for simplicity that there are two sets of operators $\{v_k\}$ and $\{h_k\}$ composed of $D=W_v+ W_h+ W_{\rm int}$ terms such that

\begin{align}
H_v &= \sum_{k=1}^{W_v} \theta_k v_k \otimes I,\qquad &H_h = \sum_{k=W_v+1}^{W_v +W_h} \theta_k I \otimes h_k \nonumber\\
H_{\rm int} &= \sum_{k=W_v+W_h +1}^{W_v +W_h + W_{\rm int}} \theta_k v_k \otimes h_k,&[h_k,h_j]  = 0~\forall~j,k,\label{eq:assumptions}
\end{align}
which implies that the Hamiltonian can in general be expressed as
\begin{equation}
\label{eq:ham_form}
H=\sum_{k=1}^D \theta_k v_k \otimes h_k.
\end{equation}
We break up the Hamiltonian into this form to emphasize the qualitative difference between the types of terms that can appear in this model.
Note that we generally assume throughout this article that $v_k,h_k$ are unitary operators, which is typically the case.

The intention of the form of the Hamiltonian in~\eqref{eq:assumptions} is to force the non-commuting terms, i.e., terms for which it holds that the commutator $[v_k,h_k]\neq 0$, to act only on the visible units of the model.
In contrast, only commuting Hamiltonian terms act on the hidden register.
Since the hidden units commute, the eigenvalues and eigenvectors for the Hamiltonian can be expressed as
\begin{equation}
H \ket{v_h} \otimes \ket{h} = \lambda_{v_h, h} \ket{v_h} \otimes \ket{h},
\end{equation}
where both the conditional eigenvectors and eigenvalues for the visible subsystem are functions of the eigenvector $\ket{h}$ obtained in the hidden register, and we hence denote these as $v_h,\lambda_{v_h,h}$ respectively.
This allows the hidden units to select between eigenbases to interpret the input data while also penalizing portions of the accessible Hilbert space that are not supported by the training data.
However, since the hidden units commute they cannot be used to construct a non-diagonal eigenbasis.\
This division of labor between the visible and hidden layers not only helps build intuition about the model but also opens up the possibility for more efficient training algorithms that exploit this fact.

For the first result we rely on a variational bound on the entropy in order to train the quantum Boltzmann machine weights for a Hamiltonian $H$ of the form given in \eqref{eq:ham_form}.
We can express this variational bound compactly in terms of a thermal expectation against a fictitious thermal probability distribution.
We define this expectation below.
\begin{definition}
\label{def:prob_distr}
Let $\tilde{H}_h = \sum_k \theta_k \tr{\rho v_k} h_k$ be the Hamiltonian acting conditioned on the visible subspace only on the hidden subsystem of the Hamiltonian $H:=\sum_k \theta_k v_k \otimes h_k$.
Then we define the expectation value over the marginal distribution over the hidden variables $h$ as
	\begin{equation}
	\label{eq:def_exp}
	\mathbb{E}_h(\cdot) =  \sum_h   \frac{(\cdot) e^{- \tr{\rho \tilde{H}_h }}}{\sum_h e^{- \tr{\rho \tilde{H}_h}}}.
	\end{equation}
\end{definition}

%We now use this definition to state the expression for the gradient of our variational bound on the quantum relative entropy.
%The bound is concretely given by
%\begin{align}
% \widetilde{S}:= \tr{\rho \log \rho} +  \tr{\rho \frac{\sum_h \alpha_h \sum_{k} E_{h,k} \theta_k v_k + \sum_h \alpha_h \log \alpha_h}{\sum_{h'} \alpha_{h'}}} + \log Z ,
%\end{align}

Using this we derive an upper bound on $S$ in section~\ref{sec:variational_bound_derivative} of the supplemental material, which leads to the following lemma.

\begin{lemma}
\label{def:Stilde}
    Assume that the Hamiltonian $H$ of the quantum Boltzmann machine takes the form described in eq.~\ref{eq:ham_form}, where $\theta_k$ are the parameters which determine the interaction strength and $v_k,h_k$ are unitary operators.
    Furthermore, let $h_k \ket{h}=E_{h,k}\ket{h}$ be the eigenvalues of the hidden subsystem, and  $\mathbb{E}_h(\cdot)$ as given by Definition~\ref{def:prob_distr}, i.e., the expectation value over the effective Boltzmann distribution of the visible layer with $\tilde{H}_h :=\sum_k E_{h,k} \theta_k v_k$.
    Then, a variational upper bound $\widetilde{S}$ of the objective function, meaning that $\widetilde{S}(\rho|H)\ge  S(\rho | e^{-H}/Z)$, is given by
    \begin{align}
        \label{eq:varbd}
  		 \widetilde{S}(\rho|H)  := \tr{\rho \log \rho} +  \tr{\rho { \sum_{k} \mathbb{E}_{h} \left[ E_{h,k} \theta_k v_k \right] +  \mathbb{E}_{h} \left[ \log \alpha_h\right] } } + \log Z,
  	\end{align}
    where $\alpha_h =  \frac{e^{- \tr{\rho \widetilde{H}_h }}}{\sum_h e^{- \tr{\rho \widetilde{H}_h}}}$ is the corresponding Gibbs distribution for the visible units.
\end{lemma}

The proof that~\eqref{eq:varbd} is a variational bound proceeds in two steps.
First, we note that for any probability distribution $\alpha_h$
\begin{equation}
\tr{\rho \log \left(\sum_{h=1}^N e^{-\sum_{k} E_{h,k} \theta_k v_k} \right)} = \tr{\rho \log \left(\sum_{h=1}^N \alpha_h \frac{e^{-\sum_{k} E_{h,k} \theta_k v_k}/\alpha_h }{\sum_{h'} \alpha_{h'}} \right)}
\end{equation}
We then apply Jensen's inequality and minimize the result over all $\alpha_h$.
This not only verifies that $\widetilde{S}(\rho|H) \ge S(\rho|H)$ but also yields a variational bound.
The details of the proof can be found in eq.~\ref{eq:def_variational_bound} in section~\ref{sec:variational_bound_derivative} of the supplemental material.

Using the above assumptions we can obtain the gradient of the variational upper bound of the relative entropy which is derived in the section~\ref{sec:var_grad_estimation}
of the supplemental material and summarized in lemma~\ref{lem:gradient_visible_layer}.

\begin{lemma}
\label{lem:gradient_visible_layer}
  Assume that the Hamiltonian $H$ of the quantum Boltzmann machine takes the form described in eq.~\ref{eq:ham_form}, where $\theta_k$ are the parameters which determine the interaction strength and $v_k,h_k$ are unitary operators.
  Furthermore, let $h_k \ket{h}=E_{h,k}\ket{h}$ be the eigenvalues of the hidden subsystem, and  $\mathbb{E}_h(\cdot)$ as given by Definition~\ref{def:prob_distr}, i.e., the expectation value over the effective Boltzmann distribution of the visible layer with $\tilde{H}_h :=\sum_k E_{h,k} \theta_k v_k$.
  Then, the derivatives of $\widetilde{S}$ with respect to the parameters of the Boltzmann machine are given by
    \begin{align}
       \label{eq:gradient_var_bound}
      \frac{\partial \widetilde{S}(\rho|H)}{\partial_{\theta_p}} = \mathbb{E}_{h} \left[ \tr{\rho E_{h,p} v_p}\right] - \tr{\frac{\partial H}{\partial \theta_p} \frac{e^{-H}}{Z}}.
    \end{align}
\end{lemma}

Notably, if we consider no interactions between the visible and the hidden layer, then indeed the gradient above reduces to the case of the visible Boltzmann machine, which was treated in \cite{kieferova2016tomography}, resulting in the gradient
\begin{equation}
  \tr{\rho \partial_{\theta_p}H} - \tr{\frac{e^{-H}}{Z} \partial_{\theta_p} H},
\end{equation}
under our assumption on the form of $H$, $\partial_{\theta_p}H = v_p$.\\

  From Lemma~\ref{lem:gradient_visible_layer}, we know the form of the derivatives of the relative entropy w.r.t.\ any parameter $\theta_p$ via Eq.~\ref{eq:gradient_var_bound}.
  Note that we can easily evaluate the second term by preparing the Gibbs state $\sigma_{Gibbs} := e^{-H}/Z$ and then evaluating the expectation value of the
  operator $\partial_{\theta_j}H$ w.r.t. the Gibbs state, using amplitude estimation for the Hadamard test~\cite{aharonov2009polynomial}.
  This is a standard procedure and we describe it in algorithm~\ref{alg:algo0} in section~\ref{sec:eval_dm_operator} of the supplemental material.

  The computational complexity of this procedure is easy to evaluate.
  If $T_{Gibbs}$ is the query complexity for the Gibbs state preparation, the query complexity of the whole algorithm including the phase estimation step is then given by $O(T_{Gibbs}/\epsilon)$ for an $\epsilon$-accurate estimate of phase estimation.
	 Next, we derive an algorithm to evaluate the first term, which requires a more involved process.
  For this, note first that we can evaluate each term $\tr{\rho v_k}$ independently from $\mathbb{E}_{h} \left[ E_{h,p}\right]$, and individually for all $k \in [D]$, i.e., all $D$ dimensions of the gradient.
  This can be done via the Hadamard test for $v_k$ which we recapitulate in section~\ref{sec:eval_dm_operator} of the supplemental material, assuming $v_k$ is unitary.
	More generally, for non-unitary $v_k$ we could evaluate this term using a linear combination of unitary operations.
	Therefore, the remaining task is to evaluate the terms $\mathbb{E}_{h} \left[ E_{h,p}\right]$ in \eqref{eq:gradient_var_bound},
  which reduces to sampling elements according to the distribution $\{\alpha_h\}$, recalling that $h_p$ applied to the subsystem has eigenvalues $E_{h,p}$.
	For this we need to be able to create a Gibbs distribution for the effective Hamiltonian $\tilde{H}_h = \sum_k \theta_k \tr{\rho v_k} h_k$ which contains only $D$ terms and can hence be evaluated efficiently as long as $D$ is small, which we can generally assume to be true.   	In order to sample according to the distribution $\{\alpha_h\}$, we first evaluate the factors $\theta_k \tr{\rho v_k}$ in the sum over $k$ via the Hadamard test, and then use these in order to implement the Gibbs distribution $\exp{(-\tilde{H_h})}/\tilde{Z}$ for the Hamiltonian $$\tilde{H}_h = \sum_k \theta_k \tr{\rho v_k} h_k.$$
  The algorithm is summarized in Algorithm~\ref{alg:algo1}.

  \begin{algorithm}[tb!]
  \caption{Variational gradient estimation - term 1}
  \label{alg:algo1}
  \begin{algorithmic}
    \STATE {\bfseries Input:} An upper bound $\tilde{S}(\rho | H )$ on the quantum relative entropy, density matrix $\rho \in \mathbb{C}^{2^n\times 2^n}$, and Hamiltonian $H \in \mathbb{C}^{2^n\times 2^n}$.
    \STATE {\bfseries Output:} Estimate $\mathcal{S}$ of the gradient $\nabla \tilde{S}$ which fulfills Thm.~\ref{thm:gradient_results}.
    \STATE {\bfseries 1.} Use Gibbs state preparation to create the Gibbs distribution for the effective Hamiltonian $\tilde{H}_h = \sum_k \theta_k \tr{\rho v_k} h_k$ with sparsity $d$.
    \STATE {\bfseries 2.} Prepare a Hadamard test state, i.e., prepare an ancilla qubit in the $\ket{+}$-state and apply a controlled-$h_k$ conditioned on the ancilla register, followed by a Hadamard gate, i.e.,
    \begin{align}
      \ket{\phi} := \frac{1}{2} \left( \ket{0} \left(\ket{\psi}_{Gibbs} + (h_k \otimes I) \ket{\psi}_{Gibbs} \right) + \ket{1} \left( \ket{\psi}_{Gibbs} -  (h_k \otimes I) \ket{\psi}_{Gibbs} \right) \right)
    \end{align}
    where $\ket{\psi}_{Gibbs} := \sum_{h} \frac{e^{-E_h/2}}{\sqrt{Z}} \ket{h}_A \ket{\phi_h}_B$ is the purified Gibbs state.
    \STATE {\bfseries 3.} Perform amplitude estimation on the $\ket{0}$ state,we need to implement the amplitude estimation with reflector $P:= -2 \ket{0}\bra{0} +I$,
      and operator $G:= \left( 2 \ket{\phi}\bra{\phi} - I \right) (P \otimes I)$.
    \STATE {\bfseries 4.} Measure now the phase estimation register which returns an $\tilde \epsilon$-estimate of the probability $\frac{1}{2} \left( 1 + \mathbb{E}_h[E_{h,k}] \right)$ of the Hadamard test to return $0$
    \STATE {\bfseries 6.} Repeat the procedure for all $D$ terms and output the first term of $\nabla \tilde{S}$.
  \end{algorithmic}
  \end{algorithm}

  The algorithm is build on three main subroutines.
  The first one is Gibbs state preparation, which is a known routine which we recapitulate in Theorem~\ref{thm:Gibbs_state_prep} in the supplemental material.
  The two remaining routines are the Hadamard test and amplitude estimation, both are well established quantum algorithms.
  The Hadamard test, will allow us to estimate the probability of the outcome.
  This is concretely given by
  \begin{equation}
            \mathrm{Pr}(0) = \frac{1}{2} \left(1 + \mathrm{Re}\bra{\psi}_{Gibbs} (h_k \otimes I) \ket{\psi}_{Gibbs} \right) =  \frac{1}{2} \left(1 + \sum_h \frac{e^{-E_h} E_{h,k}}{Z} \right),
	\end{equation}
	i.e., from $\mathrm{Pr}(0)$ we can easily infer the estimate of $\mathbb{E}_h\left[ E_{h,k} \right]$ up to precision $\epsilon$ for all the $k$ terms,
  since the last part is equivalent to $\frac{1}{2} \left( 1 + \mathbb{E}_h[E_{h,k}] \right)$.
  To speed up the time for the evaluation of the probability $\mathrm{Pr}(0)$, we use amplitude estimation.
  We recapitulate this procedure in detail in the suppemental material in section~\ref{sec:amplitude_estimation}.
  In this case, we let $P:= -2 \ket{0}\bra{0} +I$ be the reflector, where $I$ is the identity which is just the Pauli $z$ matrix up to a global phase,
  and let $G:= \left( 2 \ket{\phi}\bra{\phi} - I \right) (P \otimes I)$, for $\ket{\phi}$ being the state after the Hadamard test prior to the measurement.
	The operator $G$ has then the eigenvalue $\mu_{\pm}= \pm e^{\pm i 2 \theta}$ , where $2 \theta = \arcsin{\sqrt{\mathrm{Pr}(0)}}$,
  and $\mathrm{Pr}(0)$ is the probability to measure the ancilla qubit in the $\ket{0}$ state.
	Let now $T_{Gibbs}$ be the query complexity for preparing the purified Gibbs state (c.f. eq~\eqref{eq:gibbs_state_query_complexity} in the supplemental material).
	We can then perform phase estimation with precision $\epsilon$ for the operator $G$ requiring $O(T_{Gibbs}/\tilde \epsilon)$ queries to the oracle of $H$.

  In section~\ref{app:proof_thm_gradient_results} of the supplemental material we analyse the runtime and error of the above algorithm.
  The result is summarized in Theorem~\ref{thm:gradient_results}.

  \begin{theorem}
  \label{thm:gradient_results}
  Assume that the Hamiltonian $H$ of the quantum Boltzmann machine takes the form described in eq.~\ref{eq:ham_form}, where $\theta_k$ are the parameters which determine the interaction strength and $v_k,h_k$ are unitary operators.
  Furthermore, let $h_k \ket{h}=E_{h,k}\ket{h}$ be the eigenvalues of the hidden subsystem, and  $\mathbb{E}_h(\cdot)$ as given by Definition~\ref{def:prob_distr},
  i.e., the expectation value over the effective Boltzmann distribution of the visible layer with $\tilde{H}_h :=\sum_k E_{h,k} \theta_k v_k$,
  and suppose that $I \preceq \tilde{H}_h$ with bounded spectral norm $\norm{\tilde{H}_h(\theta)} \leq \norm{\theta}_1$, and let $\tilde{H}_h$ be $d$-sparse.
  Then $\mathcal{S}\in \mathbb{R}^D$ can be computed for any $\epsilon \in (0,\max\{1/3, 4\max_{h,p} |E_{h,p}|\})$ such that
  \begin{equation}
  	\norm{\mathcal{S} - \nabla \tilde{S}}_{max}  \leq \epsilon,
  \end{equation}
  with
  \begin{equation}
  \label{eq:hidden_gradient_query_complexity}
    \widetilde{\mathcal{O}} \left(\sqrt{\xi}\frac{D \norm{\theta}_1 dn^2 }{\epsilon} \right),
  \end{equation}
  queries to the oracle $O_H$ and $O_{\rho}$ with probability at least $2/3$, where $\norm{\theta}_1$ is the sum of absolute values of the parameters of the Hamiltonian,
   $\xi := \max[N/z, N_h/z_h]$, $N=2^n$, $N_h=2^{n_h}$, and $z,z_h$ are known lower bounds on the partition functions for the Gibbs state of $H$ and $\tilde{H}_h$ respectively.
  \end{theorem}

Theorem~\ref{thm:gradient_results} shows that the computational complexity of estimating the gradient grows the closer we get to a pure state, since for a pure state the inverse temperature $\beta\rightarrow \infty$, and therefore the norm $\norm{H(\theta)}\rightarrow \infty$, as the Hamiltonian is depending on the parameters, and hence the type of state we describe. In such cases we typically would rely on alternative techniques.
However, this cannot be generically improved because otherwise we would be able to find minimum energy configurations using a number of queries in $o(\sqrt{N})$, which would violate lower bounds for Grover's search.
Therefore more precise statements of the complexity will require further restrictions on the classes of problem Hamiltonians to avoid lower bounds imposed by Grover's search and similar algorithms.

\subsection{Gradient based training for general Hamiltonians}

Our second scheme to train a quantum Boltzmann machine is general applicable and does not require a particular form of the Hamiltonian as was required for the first approach.
We use higher order divided difference estimates for the relative entropy error based on function approximation schemes.
For this we generate differentiation formulas by differentiating an interpolant.
The idea for this is straightforward: First we construct an interpolating polynomial from the data.
Second, an approximation of the derivative at any point is obtained by a direct differentiation of the interpolant.
Concretely we perform the following steps.
We first approximate the logarithm via a Fourier-like approximation, i.e., $\log \sigma_v \rightarrow \log_{K,M}\sigma_v,$ where the subscripts $K,M$
indicate the level of truncation similar to~\cite{van2017quantum}.
This will yield a Fourier-like series in terms of $\sigma_v$, i.e., $\sum_m c_m \exp{(im\pi \sigma_v)}$.\\
Next, we need to evaluate the gradient of the function $\tr{ \frac{\partial}{\partial \theta} \rho \log_{K,M}(\sigma_v)}$.
Taking the derivative yields many terms of the form
\begin{equation}
  \label{eq:integral_partial_derivatives}
  \int_0^1 ds e^{(ism\pi \sigma_v)} \frac{\partial \sigma_v}{\partial \theta} e^{(i(1-s)m\pi \sigma_v)},
\end{equation}
as a result of the Duhamel's formula for the derivative of exponentials of operators (c.f., Sec.~\ref{eq:operator_log_ineq} of the supplemental material).
Each term in this expansion can furthermore be evaluated separately via a sampling procedure, since the terms in Eq.~\ref{eq:integral_partial_derivatives} can be approximated by $\mathbb{E}_s \left[ e^{(ism\pi \sigma_v)} \frac{\partial \sigma_v}{\partial \theta} e^{(i(1-s)m\pi \sigma_v)} \right]$.
Furthermore, since we only have a logarithmic number of terms, we can combine the results of the individual terms via classical postprocessing once we have evaluated the trace.\\
Now, we apply a divided difference scheme to approximate the gradient term $\frac{\partial \sigma_v}{\partial \theta}$ which results
in an interpolation polynomial $\mathcal{L}_{\mu,j}$ of order $l$ (for $l$ being the number of points at which we evaluate the function) in $\sigma_v$ which we can efficiently evaluate.\\
However, evaluating these terms is still not trivial. The final step consists hence of implementing a routine which allows us to evaluate
these terms on a quantum device. In order to do so, we once again make use of the Fourier series approach.
This time we take the simple idea of aproximating the density operator $\sigma_v$ by the series of itself,
i.e., $\sigma_v \approx F(\sigma_v) := \sum_{m'} c_{m'} \exp{(im \pi m' \sigma_v)}$, which we can implement conveniently via sample based Hamiltonian simulation~\cite{lloyd2014quantum,kimmel2017hamiltonian}.\\
Following these steps we obtain the expression in Eq.~\ref{eq:full_approximation}.
The real part of
\begin{equation}
  \label{eq:full_approximation}
  \sum_{m=-M_1}^{M_1} \sum_{m'=-M_2}^{M_2} \frac{i c_m \tilde{c}_{m'} m \pi}{2} \sum_{j=0}^{\mu} \mathcal{L}'_{\mu,j}(\theta) \mathbb{E}_{s \in[0,1]} \left[\tr{ \rho e^{\frac{i s \pi m}{2}\sigma_v} e^{\frac{i \pi m'}{2} \sigma_v(\theta_j)} e^{\frac{i (1-s) \pi m}{2}\sigma_v}} \right].
\end{equation}
 then approximates $\partial_{\theta} \tr{\rho \log \sigma_v}$ with at most $\epsilon$ error, where $\mathcal{L}'_{\mu,j}$ is the derivative of the interpolation polynomial which we obtain using divided differences,
and $\{c_i\}_{i},\{\tilde{c}_j\}_{j}$ are coefficients of the approximation polynomials, which can efficiently be evaluated classically.
We can evaluate each term in the sum separately and combine the results then via classical post-processing, i.e., by using the quantum computer to evaluate terms containing the trace.

\begin{algorithm}[tb!]
\caption{Gradient estimation via series approximations}
\label{alg:algo2}
\begin{algorithmic}
  \STATE {\bfseries Input:} Density matrices $\rho \in \mathbb{C}^{2^n \times 2^n}$ and $\sigma_v \in \mathbb{C}^{2^{n_v}\times 2^{n_v}}$, precalculated parameters $K,M$ and Fourier-like series for the gradient as described in eq.~\ref{eq:full_approximation}.
  \STATE {\bfseries Output:} Estimate $\mathcal{G}$ of the gradient $\nabla_{\theta}\tr{\rho \log \sigma_v}$ with guarantees in Thm.~\ref{thm:gradient_results}.
  \STATE {\bfseries 1.} Prepare the $\ket{+} \otimes \rho$ state for the Hadamard test.
  \STATE {\bfseries 2.} Conditionally on the first qubit apply sample based Hamiltonian simulation to $\rho$, i.e.,
  for $U:=e^{\frac{i s \pi m}{2}\sigma_v} e^{\frac{i \pi m'}{2} \sigma_v(\theta_j)} e^{\frac{i (1-s) \pi m}{2}\sigma_v}$, apply $\ket{0}\bra{0} \otimes I + \ket{1}{1} \otimes U$.
  \STATE {\bfseries 3.} Apply another Hadamard gate to the first qubit.
  \STATE {\bfseries 4.} Repeat the above procedure and measure the final state each time and return the averaged output.
\end{algorithmic}
\end{algorithm}

The main challenge for the algorithmic evaluation hence to compute the terms
\begin{equation}
  \label{eq:error_formula_sample_based_ham_sim}
	\tr{ \rho e^{\frac{i s \pi m}{2}\sigma_v} e^{\frac{i \pi m'}{2} \sigma_v(\theta_j)} e^{\frac{i (1-s) \pi m}{2}\sigma_v}}.
\end{equation}
Evaluating this expression is done through Algorithm~\ref{alg:algo2}, relies on two established subroutines, namely sample based Hamiltonian simulation~\cite{lloyd2014quantum,kimmel2017hamiltonian},
and the Hadamard test.
Note that the sample based Hamiltonian simulation approach introduces an additional $\epsilon_h$-error in trace norm, which we also need to take into account in the analysis.
In section~\ref{app:proof_general_algo} of the supplemental material we derive the following guarantees for Algorithm.~\ref{alg:algo2}.

  \begin{theorem}
    \label{thm:complexity_general_algo}
    Let $\rho, \sigma_v$ being two density matrices, $\norm{\sigma_v} <1/\pi$, and we have access to an oracle $O_{H}$ that computes the locations of non-zero matrix elements in each row and their values for the $d$-sparse
    Hamiltonian $H(\theta)$ (as per~\cite{berry2007efficient}) and an oracle $O_{\rho}$ which returns copies of  purified density matrix of the data $\rho$, and $\epsilon \in (0,1/6)$ an error parameter.
    With probability at least $2/3$ we can obtain an estimate $\mathcal{G}$ of the gradient w.r.t. $\theta \in \mathbb{R}^D$ of the relative entropy
    $\nabla_{\theta} \tr{\rho \log \sigma_v}$ such that
    \begin{equation}
      \norm{\nabla_{\theta}\tr{\rho \log \sigma_v} - \mathcal{G}}_{max} \leq \epsilon,
    \end{equation}
    with
  \begin{align}
\label{eq:final_query_complexity}
     \tilde{O} \left(   \sqrt{\frac{N}{z}}
      \frac{D \norm{H(\theta)}
        d   \mu^5
      \gamma
      }
      {\epsilon^3}
    \right),
  \end{align}
  queries to $O_H$ and $O_{\rho}$, where $\mu\in O(n_h + \log(1/\epsilon))$, $\|\partial_{\theta} \sigma_v\| \le e^\gamma$, $\norm{\sigma_v}\geq 2^{-n_v}$ for $n_v$
  being the number of visible units and $n_h$ being the number of hidden units, and $$\tilde{O}\left(\text{poly}\left(\gamma, n_v, n_h,\log(1/\epsilon)\right)\right)$$
  classical precomputation.
  \end{theorem}
In order to obtain the bounds in Theorem~\ref{thm:complexity_general_algo} we decompose the total error into the errors that we incur at each step of the approximation scheme,
\begin{align}
\label{eq:bounds_approximation_error}
  &\left\lvert \partial_{\theta}\tr{\rho \log \sigma_v} - \partial_{\theta} \tr{\rho \log_{K_1,M_1}^s \tilde{\sigma}_v} \right\rvert \leq \sum_i \sigma_i(\rho) \cdot \left\lVert \partial_{\theta} [\log \sigma_v - \log_{K_1,M_1}^s \tilde{\sigma}_v] \right\rVert \nonumber \\
    &\leq \sum_i \sigma_i(\rho) \cdot \left( \left\lVert \partial_{\theta} [\log \sigma_v - \log_{K_1,M_1} \sigma_v] \right\rVert \right. \nonumber \\
    &+ \left.
    \left\lVert \partial_{\theta} [\log_{K_1,M_1} \sigma_v - \log_{K_1,M_1} \tilde{\sigma}_v] \right\rVert + \left\lVert \partial_{\theta} [\log_{K_1,M_1} \tilde{\sigma_v} - \log^s_{K_1,M_1} \tilde{\sigma}_v] \right\rVert \right).
\end{align}
Then bounding each term separately and adjusting the parameters to obtain an overall error of $\epsilon$ allows us to obtain the above result.
We are hence able to use this procedure to efficiently obtain gradient estimates for a QBM with hidden units, while making minimal assumptions on the input data.

\section{Conclusion}
Generative models play an important role in quantum computing as they yield concise models for complex quantum states that have no known a priori structure.  In this article, we solve an outstanding problem in the field: the previous inability to train quantum generative models to minimize the quantum relative entropy (the analogue of the KL-divergence) between the input training set and the output quantum distribution for quantum devices with hidden units.  The inability to handle hidden units, for models such as the quantum Boltzmann machine, was a substantial drawback.

Our work showed, given an efficient subroutine for preparing Gibbs states and an efficient algorithm for computing the matrix elements of the Hamiltonian, that one can efficiently train a quantum Boltzmann machine.  Specifically, we provide two quantum algorithms for training the devices.  The first assumes that the Hamiltonian terms acting on the Hidden units are mutually commuting and relies on optimizing a variational bound on the relative entropy; whereas the second method is completely general and is based on quantum finite difference methods and Fourier techniques.
In fact, this approach is sufficiently general that similar ideas could be used to train models where the input data density operator is not a thermal state.
This would allow much more general, and therefore potentially more powerful, models to be trained without necessitating Gibbs states preparation.
We observe that the first training method requires requires polynomially fewer queries in both the number of units and the error tolerance, making it much more practical but much less general.

A number of open problems remain.
First, while we show upper bounds on the query complexity for training Boltzmann machines lower bounds have not been demonstrated.
Consequently, we do not know whether linear scaling in $\|\theta\|_1$ is optimal, as it is in Hamiltonian simulation.
However, linear scaling in $D$ is unlikely to be optimal because of recent results on quantum gradient descent which only require $\widetilde{O}(\sqrt{D})$~\cite{gilyen2019optimizing} complexity.

A related issue surrounding this work involves the complexity of performing the Gibbs state preparation.
While the method we propose in the text scales as $\widetilde{O}(\sqrt{N/Z})$~\cite{chowdhury2016quantum}, other methods exist that potentially yield better scaling in certain circumstances~\cite{yung2012quantum,van2017quantum}.
Continuing to find better methods for preparing Gibbs states will likely be a vital task to make the training of QBMs practical on near-term quantum devices.
While they do not come with theoretical bounds, recent heuristic approaches that are inspired by ideas from quantum thermodynamics and other physical phenomena may be useful to make the constant factors involved in the state preparation process palatable.

By including all these optimizations it is our hope that quantum Boltzmann machines may not be just a theoretical tool that can one day be used to model quantum states but rather an experimental tool that will be useful for modeling quantum or classical data sets in near term quantum hardware.

% \subsubsection*{Acknowledgements}
% We thank Guang Hao Low for helpful discussions.

\bibliography{bibliography}
\bibliographystyle{unsrt}

\section{Supplemental material}
\label{sec:appendix}

\subsection{Mathematical preliminaries}
While computing the gradient of the average log-likelihood is a straight forward task when training ordinary Boltzmann machines, finding the gradient of the quantum relative entropy is much harder.  The reason for this is that in general $[\partial_\theta H(\theta) , H(\theta)]\ne 0$.  This means that the ordinary rules that are commonly used in calculus for finding the derivative no longer hold.  One important example that we will use repeatedly is Duhamel's formula:
\begin{equation}
\partial_\theta e^{H(\theta)} = \int_{0}^1 \mathrm{d}s e^{H(\theta) s} \partial_\theta H(\theta) e^{H(\theta) (1-s)}.
\end{equation}
This formula can be easily proven by expanding the operator exponential in a Trotter-Suzuki expansion with $r$ time-slices, differentiating the result and then taking the limit as $r\rightarrow \infty$.  However, the relative complexity of this expression compared to what would be expected from the product rule serves as an important reminder that computing the gradient is not a trivial exercise.  A similar formula also exists for the logarithm as shown in Appendix~\ref{sec:appendix}.

Similarly, because we are working with functions of matrices here we need to also work with a notion of monotonicity. We will see that for some of our approximations to hold we will also need to define a notion of concavity (in order to use Jensen's inequality).  These notions are defined below.
    \begin{definition}[Operator monoticity]
      A function $f$ is operator monotone with respect to the semidefinite order if $0 \preceq A \preceq B$, for two symmetric positive definite operators implies, $f(A) \preceq f(B)$.
      A function is operator concave w.r.t. the semidefinite order if $cf(A)+(1-c)f(B) \preceq f(cA+(1-c)B)$, for all positive definite $A,B$ and $c \in [0,1]$.
    \end{definition}

We now derive or review some preliminary equations which we will need in order to obtain a useful bound on the gradients in the main work.

\begin{claim}
  Let $A(\theta)$ be a linear operator which depends linearly on the density matrix $\sigma$.
  Then
  \begin{equation}
    \label{eq:grad_inverse}
    \frac{\partial}{\partial \theta} A(\theta)^{-1} = - A^{-1} \frac{\partial \sigma}{\partial \theta} A^{-1}.
  \end{equation}
\end{claim}
\begin{proof}
  The proof follows straight forward by using the identity $I$.
  \begin{equation*}
    \gradtheta{I} = 0 = \frac{\partial}{\partial \theta} AA^{-1} = \left(\gradtheta{A}\right) A^{-1} + A \left(\gradtheta{A^{-1}}\right).
  \end{equation*}
  Reordering the terms completes the proof. This can equally be proven using the Gateau derivative.
\end{proof}
In the following we will furthermore rely on the following well-known inequality.
\begin{lemma}[Von Neumann Trace Inequality]
  Let $A \in \mathbb{C}^{n\times n}$ and $B\in \mathbb{C}^{n\times n}$ with singular values $\{\sigma_i(A)\}_{i=1}^n$ and $\{\sigma_i(B)\}_{i=1}^n$ respectively such that $\sigma_i(\cdot) \le \sigma_j(\cdot)$ if $i\le j$.
  It then holds that
  \begin{equation}
    \abs{\tr{AB}} \leq \sum_{i=1}^n \sigma(A)_i \sigma(B)_i.
  \end{equation}
\end{lemma}
Note that from this we immediately obtain
\begin{equation}
  \label{eq:trace_inequality}
  \abs{\tr{AB}} \leq \sum_{i=1}^n \sigma(A)_i \sigma(B)_i \leq \sigma_{max}(B) \sum_i \sigma(A)_i = \norm{B} \sum_i \sigma(A)_i.
\end{equation}
This is particularly useful if $A$ is Hermitian and PSD, since this implies $\abs{\tr{AB}} \leq \norm{B} \tr{A}$ for Hermitian $A$.

Since we are dealing with operators, the common chain rule of differentiation does not hold generally. Indeed the chain rule is a special case if the derivative of the operator commutes with the operator itself.
Since we are encountering a term of the form $\log \sigma(\theta)$, we can not assume that $[\sigma, \sigma']=0$, where $\sigma':= \sigma^{(1)}$ is the derivative w.r.t., $\theta$.
For this case we need the following identity similarly to Duhamels formula in the derivation of the gradient for the purely-visible-units Boltzmann machine.
\begin{lemma}[Derivative of matrix logarithm \cite{haber2018notes}]
  \label{eq:operator_log_ineq}
  \begin{equation}
    \frac{d}{dt} \log{A(t)} = \int\limits_0^1 [sA + (1-s) I]^{-1} \frac{dA}{dt} [sA + (1-s)I]^{-1}.
  \end{equation}
\end{lemma}
For completeness we here inlude a proof of the above identity.
\begin{proof}
  We use the integral definition of the logarithm~\cite{higham2008functions} for a complex, invertible, $n\times n$ matrix $A=A(t)$ with no real negative
  \begin{equation}
    \log{A} = (A-I) \int_0^1 [s(A-I)+I]^{-1}.
  \end{equation}
  From this we obtain the derivative $$\frac{d}{dt} log{A} = \frac{dA}{dt} \int_0^1 ds [s(A-I)+I]^{-1} + (A-I) \int_0^1 ds \frac{d}{dt}[s(A-I)+I]^{-1}.$$
  Applying \eqref{eq:grad_inverse} to the second term on the right hand side yields
  $$\frac{d}{dt} log{A} = \frac{dA}{dt} \int_0^1 ds [s(A-I)+I]^{-1} + (A-I) \int_0^1 ds [s(A-I)+I]^{-1} s \frac{dA}{dt}[s(A-I)+I]^{-1},$$
  which can be rewritten as
  \begin{eqnarray}
    \frac{d}{dt} log{A} =  \int_0^1 ds [s(A-I)+I][s(A-I)+I]^{-1}\frac{dA}{dt} [s(A-I)+I]^{-1} \\+ (A-I) \int_0^1 ds [s(A-I)+I]^{-1} s \frac{dA}{dt}[s(A-I)+I]^{-1},
  \end{eqnarray}
  by adding the identity $I = [s(A-I)+I][s(A-I)+I]^{-1}$ in the first integral and reordering commuting terms (i.e., $s$).
  Notice that we can hence just substract the first two terms in the integral which yields \eqref{eq:operator_log_ineq} as desired.
\end{proof}

\subsection{Amplitude estimation}
\label{sec:amplitude_estimation}
In the following we describe the established \textit{amplitude estimation algorithm}~\cite{brassard2002quantum}:

\begin{algorithm}[tb!]
\caption{Amplitude estimation}
\label{alg:algo3}
\begin{algorithmic}
  \STATE {\bfseries Input:} Density matrix $\rho$, unitary operator $U:\mathbb{C}^{2^n} \rightarrow \mathbb{C}^{2^n}$, qubit registers $\ket{0}\otimes \ket{0}^{\otimes n}$.
  \STATE {\bfseries Output:} An $\tilde \epsilon$ close estimate of $\tr{U \rho}$.
  \STATE {\bfseries 1.} Initialize two registers of appropriate sizes to the state $\ket{0} \mathcal{A} \ket{0}$, where $\mathcal{A}$ is a unitary transformation which prepares the input state, i.e., $\ket{\psi}=\mathcal{A}\ket{0}$.
  \STATE {\bfseries 2.} Apply the quantum Fourier transform $\mathrm{QFT}_N: \ket{x}\rightarrow \frac{1}{\sqrt{N}} \sum_{y=0}^{N-1} e^{2 \pi i x y/N} \ket{y}$ for $0 \leq x < N$, to the first register.
  \STATE {\bfseries 3.} Apply $\Lambda_N(Q)$ to the second register, i.e., let $\Lambda_N(U): \ket{j}\ket{y}\rightarrow \ket{j}(U^j \ket{y})$ for $0\leq j < N$,
                      then we apply $\Lambda_N(Q)$ where $Q:= -\mathcal{A}S_0 \mathcal{A}^{\dagger}S_t$ is the Grover's operator.
  \STATE {\bfseries 4.} Apply $\mathrm{QFT}^{\dagger}_N$ to the first register.
  \STATE {\bfseries 5.} Return $\tilde{a} = \sin^2(\pi \frac{\tilde \theta}{N})$.
\end{algorithmic}
\end{algorithm}

Algorithm~\ref{alg:algo3} describes the amplitude estimation algorithm.
The output is an $\epsilon$-close estimate of the target amplitude.
Note that in step (3), $S_0$ changes the sign of the amplitude if and only if the state is the zero state $\ket{0}$,
and $S_t$ is the sign-flip operator for the target state, i.e., if $\ket{x}$ is the desired outcome, then $S_t := I - 2\ket{x}\bra{x}$.

The algorithm can be summarized as the unitary transformation $$ \left( (\mathrm{QFT}^{\dagger} \otimes I) \Lambda_N(Q) (\mathrm{QFT}_N \otimes I)\right)$$
applied to the state $\ket{0}\mathcal{A}\ket{0}$, followed by a measurement of the first register and classical post-processing returns an estimate $\tilde \theta$ of the amplitude of the desired outcome such that $\lvert \theta - \tilde \theta \rvert \leq \epsilon$ with probability at least $8/\pi^2$.
The result is summarized in the following theorem, which states a slightly more general version.
\begin{theorem}[Amplitude Estimation~\cite{brassard2002quantum}]
\label{thm:amplitude_estimation}
For any positive integer $k$, the Amplitude Estimation Algorithm returns an estimate $\tilde a$ ($0\leq \tilde a \leq 1$) such that $$\lvert \tilde a - a \rvert \leq 2 \pi k \frac{\sqrt{a(1-a)}}{N} + k^2 \frac{\pi^2}{N^2}$$
with probability at least $\frac{8}{\pi^2}\approx 0.81$ for $k=1$ and with probability greater than $1-\frac{1}{2(k-1)}$ for $k\geq 2$. If $a=0$ then $\tilde a=0$ with certainty, and and if $a=1$ and $N$ is even, then $\tilde{a}=1$ with certainty.
\end{theorem}
Notice that the amplitude $\theta$ can hence be recovered via the relation $\theta = \arcsin{\sqrt{\theta_a}}$ as described above which incurs an $\epsilon$-error for $\theta$ (c.f., Lemma 7, \cite{brassard2002quantum}).

\subsection{The Hadamard test}
\label{sec:eval_dm_operator}

Here we present an easy subroutine to evaluate the trace of products of unitary operators $U$ with a density matrix $\rho$, which is known as the Hadamard test.

\begin{algorithm}[tb!]
\caption{Variational gradient estimation - term 2}
\label{alg:algo0}
\begin{algorithmic}
  \STATE {\bfseries Input:} Density matrix $\rho$, unitary operator $U:\mathbb{C}^{2^n} \rightarrow \mathbb{C}^{2^n}$, qubit registers $\ket{0}\otimes \ket{0}^{\otimes n}$.
  \STATE {\bfseries Output:} An $\tilde \epsilon$ close estimate of $\tr{U \rho}$.
  \STATE {\bfseries 1.} Prepare the first qubits $\ket{+}$ state and initialize the second register to $0$.
  \STATE {\bfseries 2.} Use an appropriate subroutine to prepare the density matrix $\rho$ on the second register to obtain the state $\ket{+}\bra{+}\otimes \rho$.
  \STATE {\bfseries 3.} Apply a controlled operation $\ket{0}\bra{0} \otimes I_{2^n} + \ket{1}\bra{1} \otimes U$, followed by a Hadamard gate.
  \STATE {\bfseries 4.} Perform amplitude estimation on the $\ket{0}$ state, via the reflector $P:= -2 \ket{0}\bra{0} +I$, and operator $G:= \left( 2 \rho - I \right) (P \otimes I)$.
  \STATE {\bfseries 5.} Measure now the phase estimation register which returns an $\tilde \epsilon$-estimate of the probability $\frac{1}{2} \left( 1 + \mathrm{Re}\left[\tr{U\rho}\right] \right) $ of the Hadamard test to return $0$.
  \STATE {\bfseries 6.} Repeat the procedure for an additional controled application of $\exp(i\pi/2)$ in step (3) to recover also the imaginary part of the result.
  \STATE {\bfseries 7.} Return the real and imaginary part of the probability estimates.
\end{algorithmic}
\end{algorithm}

Note that this procedure can easily be adapted to be used for $\rho$ being some Gibbs distribution.
We then would use a Gibbs state preparation routine in step (2).
For example for the evaluation of the gradient of the variational bound, we require this subroutine to evaluate $U=\partial_{\theta}H$ for $\rho$ being the Gibbs distribution corresponding to the Hamiltonian $H$.

\subsection{Deferred proofs}

First for convenience, we formally define quantum Boltzmann machines below.

\begin{definition}
A quantum Boltzmann machine to be a quantum mechanical system that acts on a tensor product of Hilbert spaces $\mathcal{H}_v\otimes \mathcal{H}_h \in \mathbb{C}^{2^n}$ that correspond to the visible and hidden subsystems of the Boltzmann machine.  It further has a Hamiltonian of the form $H \in \mathbb{C}^{2^n \times 2^n}$ such that $\norm{H - \text{diag}(H)} > 0$.  The quantum Boltzmann machine takes these parameters and then outputs a state of the form ${\rm Tr}_h \left(\frac{e^{-H}}{{\rm Tr}(e^{-H})} \right)$.
\end{definition}
Given this definition, we are then able to discuss the gradient of the relative entropy between the output of a quantum Boltzmann machines and the input data that it is trained with.

\subsubsection{Derivation of the variational bound}
\label{sec:variational_bound_derivative}
\begin{proof}[Proof of Lemma~\ref{def:Stilde}]
  Recall that we assume that the Hamiltonian $H$ takes the form $$H:=\sum_k \theta_k v_k \otimes h_k,$$ where $v_k$ and $h_k$ are operators acting on the visible and hidden units respectively and we can assume $h_k=d_k$ to be diagonal in the chosen basis.
  Under the assumption that $[h_i,h_j]=0, \forall i,j$, c.f. the assumptions in \eqref{eq:assumptions}, there exists a basis $\{\ket{h}\}$ for the hidden subspace such that $h_k \ket{h}=E_{h,k}\ket{h}$.
  With these assumptions we can hence reformulate the logarithm as
  \begin{align}
    &\log\trh{e^{-H}} = \log \left( \sum_{v,v',h}  \bra{v,h} e^{-\sum_{k} \theta_k v_k \otimes h_k} \ket{v',h} \ket{v}\bra{v'}\right) \\
    &= \log \left( \sum_{v,v',h}  \bra{v} e^{-\sum_{k} E_{h,k} \theta_k v_k} \ket{v'} \ket{v}\bra{v'} \right) \\
    &= \log \left(\sum_h e^{-\sum_{k} E_{h,k} \theta_k v_k} \right),
  \end{align}
  where it is important to note that $v_k$ are operators and we hence just used the matrix representation of these in the last step.
  In order to further simplify this expression, first note that each term in the sum is a positive semi-definite operator.
  In particularly, note that the matrix logarithm is operator concave and operator monotone, and hence by Jensen's inequality, for any sequence of non-negative number $\{\alpha_i\}: \sum_i \alpha_i =1$ we have that $$\log \left( \frac{\sum_{i=1}^N \alpha_i U_i}{\sum_j \alpha_j} \right) \geq \frac{\sum_{i=1}^N \alpha_i\log \left(  U_i \right)}{\sum_j \alpha_j}.$$
  and since we are optimizing $\tr{\rho \log \rho} -\tr{\rho \log \sigma_v}$ we hence obtain for arbitrary choice of $\{\alpha_i\}_i$ under the above constraints,
  \begin{align}
  \tr{\rho \log \left(\sum_{h=1}^N e^{-\sum_{k} E_{h,k} \theta_k v_k} \right)} &= \tr{\rho \log \left(\sum_{h=1}^N \alpha_h \frac{e^{-\sum_{k} E_{h,k} \theta_k v_k}/\alpha_h }{\sum_{h'} \alpha_{h'}} \right)} \nonumber \\
  & \geq -\tr{\rho \frac{\sum_h \alpha_h \sum_{k} E_{h,k} \theta_k v_k + \sum_h \alpha_h \log \alpha_h}{\sum_{h'} \alpha_{h'}}}.
  \end{align}
  Hence, the variational bound on the objective function for any $\{\alpha_i\}_i$ is
  \begin{align}
  \label{eq:def_variational_bound}
  \mathcal{O}_{\rho}(H) =& \tr{\rho \log \rho} - \tr{\rho \log \sigma_v} \nonumber \\
            &\leq \tr{\rho \log \rho} +  \tr{\rho \frac{\sum_h \alpha_h \sum_{k} E_{h,k} \theta_k v_k + \sum_h \alpha_h \log \alpha_h}{\sum_{h'} \alpha_{h'}}} + \log Z  =: \tilde{S}
  \end{align}
\end{proof}

\subsubsection{Gradient estimation}
\label{sec:var_grad_estimation}
For the following result we will rely on a variational bound in order to train the quantum Boltzmann machine weights for a Hamiltonian $H$ of the form given in \eqref{eq:ham_form}.
We begin by proving Lemma~\ref{lem:gradient_visible_layer} in the main work, which will give us an upper bound for the gradient of the relative entropy.

\begin{proof}[Proof of Lemma~\ref{lem:gradient_visible_layer}]
	We first derive the gradient of the normalization term ($Z$) in the relative entropy, which can be trivially evaluated using Duhamels formula to obtain $$\frac{\partial}{\partial \theta_p} \log \tr{e^{-H}} = - \tr{\frac{\partial H}{\partial \theta_p} \frac{e^{-H}}{Z}} =- \tr{\sigma \partial_{\theta_p} H}.$$
	Note that we can easily evaluate this term by first preparing the Gibbs state $\sigma_{Gibbs} := e^{-H}/Z$ and then evaluating the expectation value of the operator $\partial_{\theta_p}H$ w.r.t. the Gibbs state, using amplitude estimation for the Hadamard test.
	If $T_{Gibbs}$ is the query complexity for the Gibbs state preparation, the query complexity of the whole algorithm including the phase estimation step is then given by $O(T_{Gibbs}/\tilde{\epsilon})$ for an $\tilde{\epsilon}$-accurate estimate of phase estimation.
Taking into account the desired accuracy and the error propagation will hence straight forward give the computational complexity to evaluate this part.\\
	We now proceed with the gradient evaluations for the model term.
    Using the variational bound on the objective function for any $\{\alpha_i\}_i$, given in eq.~\ref{eq:def_variational_bound}, we obtain the gradient
    \begin{align}
      \frac{\partial \tilde{S}}{\partial_{\theta_p}} &=- \tr{\frac{\partial H}{\partial \theta_p} \frac{e^{-H}}{Z}} + \tr{ \frac{\partial }{\partial \theta_p} \rho \sum_h \alpha_h \sum_{k} E_{h,k} \theta_k v_k }+   \frac{\partial }{\partial \theta_p} \sum_h \alpha_h \log \alpha_h \\
&= - \tr{\frac{\partial H}{\partial \theta_p} \frac{e^{-H}}{Z}} +  \frac{\partial }{\partial \theta_p} \left( \sum_h \alpha_h \tr{ \rho \sum_{k} E_{h,k} \theta_k v_k }+  \sum_h \alpha_h \log \alpha_h \right)
    \end{align}
    where the first term results from the partition sum.
The latter term can be seen as a new effective Hamiltonian, while the latter term is the entropy.
The latter term hence resembles the free energy $F(h)=E(h)-TS(h)$, where $E(h)$ is the mean energy of the effective system with energies $E(h):=\tr{ \rho \sum_{k} E_{h,k} \theta_k v_k }$, $T$ the temperature and $S(h)$ the Shannon entropy of the $\alpha_h$ distribution. We now want to choose these $\alpha_h$ terms to minimize this variational upper bound.
It is well-established in statistical physics, see for example~\cite{landau1980statistical}, that the distribution which maximizes the free energy is the Boltzmann (or Gibbs) distribution, i.e., $$\alpha_h =  \frac{e^{- \tr{\rho \tilde{H}_h }}}{\sum_h e^{- \tr{\rho \tilde{H}_h}}},$$
where $\tilde{H}_h :=\sum_k E_{h,k} \theta_k v_k$ is a new effective Hamiltonian on the visible units, and the $\{\alpha_i\}$ are given by the corresponding Gibbs distribution for the visible units.

Therefore, our gradients can be taken with respect to this distribution and the bound above, where $\tr{\rho \tilde{H}_h}$ is the mean energy of the the effective visible system w.r.t. the data-distribution.
For the derivative of the energy term we obtain
\begin{align}
  &\frac{\partial }{\partial \theta_p}  \sum_h \alpha_h \tr{ \rho \sum_{k} E_{h,k} \theta_k v_k } = \\
  &= \sum_h \left( \alpha_h \left( \mathbb{E}_{h'}\left[\tr{ \rho E_{h',p} v_p }\right] - \tr{\rho E_{h,p} v_p} \right) \tr{\rho \tilde{H}_h} + \alpha_h \tr{\rho E_{h,p} v_p}\right) \\
  &= \mathbb{E}_{h} \left[\left( \mathbb{E}_{h'}\left[\tr{ \rho E_{h',p} v_p }\right] - \tr{\rho E_{h,p} v_p} \right) \tr{\rho \tilde{H}_h} + \tr{\rho E_{h,p} v_p}\right],
\end{align}
while the entropy term yields
\begin{align}
  \frac{\partial}{\partial \theta_p} \sum_h \alpha_h \log \alpha_h &=  \sum_h \alpha_h \left( \left[ \tr{\rho E_{h,p}v_p} - \mathbb{E}_{h'} \left[\tr{\rho E_{h',p}v_p} \right] \right] \tr{\rho \tilde{H}_h} - \tr{\rho E_{h,p} v_p} \right) \nonumber \\
  &+ \sum_h \alpha_h \left( \tr{\rho E_{h,p}} - \mathbb{E}_{h'} \left[\tr{\rho E_{h',p}v_p} \right] \right) \log  \tr{e^{-\tilde{H}_h}} \nonumber \\
  &+ \mathbb{E}_{h'} \left[\tr{\rho E_{h',p} v_p} \right].
\end{align}
This can be further simplified to
\begin{align}
  &\sum_h \alpha_h \left( \tr{\rho E_{h,p}v_p} - \mathbb{E}_{h'} \left[\tr{\rho E_{h',p}v_p} \right] \right) \tr{\rho \tilde{H}_h} \\
  =&\mathbb{E}_h \left[ \left( \tr{\rho E_{h,p}v_p} - \mathbb{E}_{h'} \left[\tr{\rho E_{h',p}v_p} \right] \right) \tr{\rho \tilde{H}_h} \right].
\end{align}
Te resulting gradient for the variational bound for the visible terms is hence given by
  \begin{align}
        \frac{\partial \tilde{S}}{\partial_{\theta_p}} &= \mathbb{E}_{h} \left[ \tr{\rho E_{h,p} v_p}\right] - \tr{\frac{\partial H}{\partial \theta_p} \frac{e^{-H}}{Z}}
  \end{align}
\end{proof}
Notably, if we consider no interactions between the visible and the hidden layer, then indeed the gradient above reduces recovers the gradient for the visible Boltzmann machine, which was treated in \cite{kieferova2016tomography}, resulting in the gradient
\begin{equation*}
  \tr{\rho \partial_{\theta_p}H} - \tr{\frac{e^{-H}}{Z} \partial_{\theta_p} H},
\end{equation*}
under our assumption on the form of $H$, $\partial_{\theta_p}H = v_p$.
% \begin{multline}
%   -\sum_v p(v) \bra{v} \frac{\partial H}{\partial \theta_k} \ket{v} + \\
%   \mathbb{E}_{\alpha} \left[ \tr{\left( \mathbb{E}_{h'}\left[\tr{ \rho E_{h',j} v_j }\right] - \tr{\rho E_{h,j} v_j} \right) \tilde{H}_h + \rho E_{h,j} v_j} +  \tr{\rho E_{h,j}v_j} - \mathbb{E}_{h'} \left[\tr{\rho E_{h',j}v_j} \right] \right) \tr{\rho \tilde{H}_h} \right].
% \end{multline}

\subsubsection{Operationalizing the gradient based training}
\label{app:proof_thm_gradient_results}
From Lemma~\ref{lem:gradient_visible_layer}, we know that the derivative of the relative entropy w.r.t.\ any parameter $\theta_p$ can be stated as
    \begin{align}
	\label{eq:gradient_approach_1_recap}
        \frac{\partial \tilde{S}}{\partial_{\theta_p}} =   \mathbb{E}_h \left[ E_{h,p} \right] \tr{\rho v_p} - \tr{\frac{\partial H}{\partial \theta_p} \frac{e^{-H}}{Z}}.
    \end{align}
    Since evaluating the latter part is, as mentioned above, straight forward, we give here an algorithm for evaluating the first part.\\
	Now note that we can evaluate each term $\tr{\rho v_k}$ individually for all $k \in [D]$, i.e., all $D$ dimensions of the gradient via the Hadamard test for $v_k$, assuming $v_k$ is unitary.
	More generally, for non-unitary $v_k$ we could evaluate this term using a linear combination of unitary operations.
	Therefore, the remaining task is to evaluate the terms $\mathbb{E}_{h} \left[ E_{h,p}\right]$ in \eqref{eq:gradient_approach_1_recap}, which reduces to sampling according to the distribution $\{\alpha_h\}$.\\
	For this we need to be able to create a Gibbs distribution for the effective Hamiltonian $\tilde{H}_h = \sum_k \theta_k \tr{\rho v_k} h_k$ which contains only $D$ terms and can hence be evaluated efficiently as long as $D$ is small, which we can generally assume to be true.   	In order to sample according to the distribution $\{\alpha_h\}$, we first evaluate the factors $\theta_k \tr{\rho v_k}$ in the sum over $k$ via the Hadamard test, and then use these in order to implement the Gibbs distribution $\exp{(-\tilde{H_h})}/\tilde{Z}$ for the Hamiltonian $$\tilde{H}_h = \sum_k \theta_k \tr{\rho v_k} h_k.$$
In order to do so, we adapt the results of \cite{van2017quantum} in order to prepare the corresponding Gibbs state (although alternative methods can also be used~\cite{poulin2009sampling,chowdhury2016quantum,yung2012quantum}).
%Improved versions of this might be possible based on the results given in \cite{gilyen2018quantum}.

 \begin{theorem}[Gibbs state preparation~\cite{van2017quantum}]
	\label{thm:Gibbs_state_prep}
    Suppose that $I \preceq H$ and we are given $K\in \mathbb{R}_+$ such that $\norm{H}\leq 2K$, and let $H \in \mathbb{C}^{N\times N}$ be a $d$-sparse Hamiltonian, and we know a lower bound $z\leq Z=\tr{e^{-H}}$.
    If $\epsilon \in (0,1/3)$, then we can prepare a purified Gibbs state $\ket{\gamma}_{AB}$ such that
    \begin{equation}
      \norm{\mathrm{Tr}_B \left[\ket{\gamma}\bra{\gamma}_{AB} \right] - \frac{e^{-H}}{Z}} \leq \epsilon
    \end{equation}
    using
    \begin{equation}
	\label{eq:gibbs_state_query_complexity}
      \tilde{\mathcal{O}} \left(\sqrt{\frac{N}{z}} Kd \log \left( \frac{K}{\epsilon}\right) \log \left( \frac{1}{\epsilon} \right) \right)
    \end{equation}
    queries, and
    \begin{equation}
      \tilde{\mathcal{O}} \left(\sqrt{\frac{N}{z}}  Kd \log \left( \frac{K}{\epsilon}\right) \log \left( \frac{1}{\epsilon} \right)  \left[ \log(N) + \log^{5/2} \left( \frac{K}{\epsilon} \right) \right] \right)
    \end{equation}
    gates.
  \end{theorem}
  Note that by using the above algorithm with $\tilde{H}_{sim}/2$, the preparation of the purified Gibbs state will leave us in the state
  \begin{equation}
    \ket{\psi}_{Gibbs} := \sum_{h} \frac{e^{-E_h/2}}{\sqrt{Z}} \ket{h}_A \ket{\phi_h}_B,
  \end{equation}
where $\ket{\phi_j}_B$ are mutually orthogonal trash states, which can typically be chosen to be $\ket{h}$, i.e., a copy of the first register, which is irrelevant for our computation, and $\ket{h}_A$ are the eigenstates of $\tilde{H}$.
Tracing out the second register will hence leave us in the corresponding Gibbs state $$\sigma_h := \sum_h \frac{e^{-E_h}}{Z} \ket{h}\bra{h}_A,$$
and we can hence now use the Hadamard test with input $h_k$ and $\sigma_h$, i.e., the operators on the hidden units and the Gibbs state, and estimate the expectation value $\mathbb{E}_h \left[E_{h,k}\right]$.
We provide such a method below.

\begin{proof}[Proof of Theorem~\ref{thm:gradient_results}]
Conceptually, we perform the following steps, starting with Gibbs state preparation followed by a Hadamard test coupled with amplitude estimation to obtain estimates of the probability of a $0$ measurement.
The proof follows straight from the algorithm described in~\ref{alg:algo1}.

From this we see that the runtime constitutes the query complexity of preparing the Gibbs state $$T_{Gibbs}^V= \tilde{\mathcal{O}} \left(\sqrt{\frac{2^n}{Z}}\frac{\norm{H(\theta)}d}{\epsilon} \log \left( \frac{\norm{H(\theta)}}{\tilde \epsilon}\right) \log \left( \frac{1}{\tilde \epsilon} \right) \right),$$ where $2^n$ is the dimension of the Hamiltonian, as given  in Theorem~\ref{thm:gradient_results} and combining it with the query complexity of the amplitude estimation procedure, i.e., $1/\epsilon$.
However, in order to obtain a final error of $\epsilon$, we will also need to account for the error in the Gibbs state preparation.
For this, note that we estimate terms of the form $\mathrm{Tr}_{AB} \left[\bra{\psi}_{Gibbs} (h_k \otimes I) \ket{\psi}_{Gibbs}^V\right] = \mathrm{Tr}_{AB} \left[ (h_k \otimes I) \ket{\psi}_{Gibbs}^V\bra{\psi}_{Gibbs}^V\right]$.
We can hence estimate the error w.r.t. the true Gibbs state $\sigma_{Gibbs}$ as
\begin{align}
\label{eq:error_prop_Gibbs_state_prep}
&\mathrm{Tr}_{AB} \left[ (h_k \otimes I) \ket{\psi}_{Gibbs}^V\bra{\psi}_{Gibbs}^V\right] -\mathrm{Tr}_{A} \left[ h_k \sigma_{Gibbs}\right]  \nonumber \\
&\qquad= \mathrm{Tr}_{A} \left[ h_k \mathrm{Tr}_{B}\left[\ket{\psi}_{Gibbs}^V\bra{\psi}_{Gibbs}^V\right] - h_k \sigma_{Gibbs}\right] \nonumber \\
&\qquad\leq \sum_i \sigma_i(h_k) \norm{\mathrm{Tr}_{B}\left[\ket{\psi}_{Gibbs}^V\bra{\psi}_{Gibbs}^V\right] - \sigma_{Gibbs}} \nonumber \\
&\qquad \leq \tilde \epsilon \sum_i \sigma_i(h_k).
\end{align}
For the final error being less then $\epsilon$, the precision we use in the phase estimation procedure, we hence need to set $\tilde \epsilon = \epsilon/(2 \sum_i \sigma_i(h_k)) \leq 2^{-n-1} \epsilon$, reminding that $h_k$ is unitary, and similarly precision $\epsilon/2$ for the amplitude estimation, which yields the query complexity of
\begin{align}
 &\mathcal{O} \left(\sqrt{\frac{N_h}{z_h}}\frac{\norm{H(\theta)}d}{\epsilon} \left(n^2+ n \log \left( \frac{\norm{H(\theta)}}{ \epsilon}\right)  + n \log \left( \frac{1}{  \epsilon}\right) + \log \left( \frac{\norm{H(\theta)}}{ \epsilon}\right) \log \left( \frac{1}{  \epsilon}\right) \right) \right),
 \nonumber\\
 &\qquad \in \widetilde{O} \left(\sqrt{\frac{N_h}{z_h}}\left(\frac{n^2 \|\theta\|_1 d}{\epsilon} \right) \right).
\end{align}
where we denote with $A$ the hidden subsystem with dimensionality $2^{n_h} \leq 2^N$, on which we want to prepare the Gibbs state and with $B$ the subsystem for the trash state.

Similarly, for the evaluation of the second part in \eqref{eq:gradient_approach_1_recap} requires the Gibbs state preparation for $H$, the Hadamard test and phase estimation.
Similar as above we meed to take into account the error.
Letting the purified version of the Gibbs state for $H$ be given by $\ket{\psi}_{Gibbs}$, which we obtain using Theorem~\ref{thm:Gibbs_state_prep}, and $\sigma_{Gibbs}$ be the perfect state, then the error is given by
\begin{align}
 &\mathrm{Tr}_{AB} \left[ (v_k \otimes h_k \otimes I) \ket{\psi}_{Gibbs} \bra{\psi}_{Gibbs} \right] -\mathrm{Tr}_{A} \left[ (v_k \otimes h_k)\sigma_{Gibbs}\right]  \nonumber \\
&\qquad= \mathrm{Tr}_{A} \left[ (v_k \otimes h_k) \mathrm{Tr}_{B}\left[\ket{\psi}_{Gibbs}^V\bra{\psi}_{Gibbs}^V\right] - (v_k \otimes h_k) \sigma_{Gibbs}\right] \nonumber \\
&\qquad\leq \sum_i \sigma_i(v_k \otimes h_k) \norm{\mathrm{Tr}_{B}\left[\ket{\psi}_{Gibbs}^V\bra{\psi}_{Gibbs}^V\right] - h_k \sigma_{Gibbs}} \nonumber \\
&\qquad \leq \tilde \epsilon \sum_i \sigma_i(v_k \otimes h_k),
\end{align}
where in this case $A$ is the subsystem of the visible and hidden subspace and $B$ the trash system.
We hence upper bound the error similar as above and introducing $\xi := \max[N/z, N_h/z_h]$ we can find a uniform bound on the query complexity for evaluating a single entry of the $D$-dimensional gradient is in
$$
 \widetilde{O} \left(\sqrt{\zeta}\left(\frac{n^2 \|\theta\|_1 d}{\epsilon} \right) \right),
$$
thus we attain the claimed query complexity by repeating this procedure for each of the $D$ components of the estimated gradient vector $\mathcal{S}$.

Note that we also need to evaluate the terms $\tr{\rho v_k}$ to precision $\hat{\epsilon} \leq \epsilon$, which though only incurs an additive cost of $D/\epsilon$ to the total query complexity, since this step is required to be performed once. Note that $|\mathbb{E}_h(h_p)|\le 1$ because $h_p$ is assumed to be unitary.
To complete the proof we only need to take the success probability of the amplitude estimation process into account.
For completeness we state the algorithm in the appendix and here refer only to Theorem~\ref{thm:amplitude_estimation}, from which we have that the procedure succeeds with probability at least $8/\pi^2$.
In order to have a failure probability of the final algorithm of less than $1/3$, we need to repeat the procedure for all $d$ dimensions of the gradient and take the median.
We can bound the number of repetitions in the following way.

Let $n_f$ be the number of instances of the gradient estimate such that the error is larger than $\epsilon$ and $n_s$ be the number of instances with an error $\leq \epsilon$ for one dimension of the gradient, and the result that we take is the median of the estimates, where we take $n=n_s+n_f$ samples.
The algorithm gives a wrong answer for each dimension if $n_s \leq \left\lfloor \frac{n}{2} \right\rfloor$, since then the median is a sample such that the error is not bound by $\epsilon$.
Let $p=8/\pi^2$ be the success probability to draw a positive sample, as is the case of the amplitude estimation procedure.
Since each instance of the phase estimation algorithm will independently return an estimate, the total failure probability is given by the union bound, i.e.,
\begin{equation}
\label{eq:probability_boost}
\pr_{fail} \leq D \cdot \pr \left[n_s \leq  \left\lfloor \frac{n}{2} \right\rfloor \right] \leq D \cdot e^{- \frac{n}{2p} \left(p - \frac{1}{2} \right)^2} \leq \frac{1}{3},
\end{equation}
which follows from the Chernoff inequality for a binomial variable with $p>1/2$, which is given in our case.
Therefore, by taking $n \geq \frac{2p}{(p-1/2)^2} \log(3D) = \frac{16}{(8-\pi^2/2)^2} \log(3D) = O(\log(3D))$, we achieve a total failure probability of at most $1/3$.

This is sufficient to demonstrate the validity of the algorithm if
\begin{equation}
  \label{eq:trace_est}
  \tr{\rho \tilde H_h}
\end{equation}
is known exactly.  This is difficult to do because the probability distribution $\alpha_h$ is not usually known apriori.  As a result, we assume that
the distribution will be learned empirically and to do so we will need to draw samples from the purified Gibbs states used as input.
This sampling procedure will incur errors.
To take such errors into account assume that we can obtain estimates $T_h$ of \ref{eq:trace_est} with precision $\delta_t$, i.e.,
\begin{equation}
  \left\lvert T_h - \tr{\rho \tilde H_h} \right\rvert \leq \delta_t.
\end{equation}
Under this assumption we can now bound the distance $\lvert \alpha_h - \tilde{\alpha}_h \rvert$ in the following way.
Observe that
\begin{align}
  \label{eq:error_step1}
  \left\lvert  \alpha_h - \tilde{\alpha}_h\right\rvert &= \left\lvert  \frac{e^{-\tr{\rho \tilde{H}_h}}}{\sum_h e^{-\tr{\rho \tilde{H}_h}}} - \frac{T_h}{\sum_h T_h} \right\rvert \nonumber \\
    &\leq \left\lvert \frac{e^{-\tr{\rho \tilde{H}_h}}}{\sum_h e^{-\tr{\rho \tilde{H}_h}}} - \frac{T_h}{\sum_h e^{-\tr{\rho \tilde{H}_h}}} \right\rvert + \left\lvert \frac{T_h}{\sum_h e^{-\tr{\rho \tilde{H}_h}}} - \frac{T_h}{\sum_h T_h} \right\rvert,
\end{align}
and we hence need to bound the following two quantities in order to bound the error.
First, we need a bound on
\begin{align}
  \label{eq:error_first_part_1}
  \left\lvert e^{-\tr{\rho \tilde H_h}} - e^{-T_h} \right\rvert.
\end{align}
For this, let $f(s):= T_h\ (1-s) + \tr{\rho \tilde H_h}\ s$, such that eq.~\ref{eq:error_first_part_1} can be rewritten as
\begin{align}
  \label{eq:error_first_part_2}
  \left\lvert e^{-f(1)} - e^{-f(0)} \right\rvert &= \left\lvert \int_0^1 \frac{d}{ds} e^{-f(s)} ds \right\rvert \nonumber \\
        &= \left\lvert \int_0^1 \dot{f}(s) e^{-f(s)} ds \right\rvert \nonumber \\
        &= \left\lvert \int_0^1 \left(\tr{\rho \tilde H_h} - T_h \right) e^{-f(s)} ds \right\rvert \nonumber \\
        &\leq \delta e^{-\min_s f(s)} \nonumber \\
        &\leq \delta e^{-\tr{\rho
         H_h} + \delta}
\end{align}
and assuming $\delta \leq \log(2)$, this reduces to
\begin{align}
  \left\lvert e^{-f(1)} - e^{-f(0)} \right\rvert \leq 2 \delta e^{-\tr{\rho \tilde H_h}}.
\end{align}
Second, we need the fact that
\begin{align}
  \left\lvert \sum_h e^{-\tr{\rho \tilde H_h}} - \sum_h T_h \right\rvert \leq 2 \delta \sum_h e^{-\tr{\rho \tilde H_h}}.
\end{align}
Using this, eq.~\ref{eq:error_step1} can be upper bound by
\begin{align}
  &\frac{2 \delta e^{-\tr{\rho \tilde H_h}}}{\sum_h e^{-\tr{\rho \tilde H_h}}} + \lvert T_h\rvert \left\lvert \frac{1}{\sum_h e^{-\tr{\rho \tilde H_h}}} - \frac{1}{(1-2\delta)\sum_h e^{-\tr{\rho \tilde H_h}}} \right\rvert \nonumber \\
      &\qquad\leq \frac{2 \delta e^{-\tr{\rho \tilde H_h}}}{\sum_h e^{-\tr{\rho \tilde H_h}}} + \frac{4 \delta \lvert T_h\rvert}{\sum_h e^{-\tr{\rho \tilde H_h}}},
\end{align}
where we used that $\delta \leq 1/4$.
Note that
\begin{align}
  4 \delta \lvert T_h \rvert &\leq 4 \delta \left(e^{-\tr{\rho \tilde H_h}} + 2 \delta e^{-\tr{\rho \tilde H_h}} \right) \nonumber \\
            &= e^{-\tr{\rho \tilde H_h}} \left(4 \delta + 8 \delta^2 \right) \nonumber \\
            &\leq e^{-\tr{\rho \tilde H_h}}(4 \delta + 2 \delta) \nonumber \\
            &\leq 6 \delta e^{-\tr{\rho \tilde H_h}},
\end{align}
which leads to a final error of
\begin{equation}
  \lvert \alpha_h - \tilde{\alpha}_h \rvert \leq 8 \delta \frac{e^{-\tr{\rho \tilde H_h}}}{\sum_h e^{-\tr{\rho \tilde H_h}}}.
\end{equation}
With this we can now bound the error in the expectation w.r.t. the faulty distribution for some function $f(h)$ to be
\begin{align}
  \left\lvert \mathbb{E}_h(f(h)) - \tilde{\mathbb{E}}_h(f(h)) \right\rvert
  &\leq 8 \delta \sum_h \frac{ f(h) e^{-\tr{\rho \tilde H_h}} }{\sum_h e^{-\tr{\rho \tilde H_h}}} \nonumber \\
        &\leq 8 \delta \max_h f(h).
\end{align}
We can hence use this in order to estimate the error introduced in the first term of eq.~\ref{eq:gradient_approach_1_recap} through errors in the distribution $\{ \alpha_h\}$ as
\begin{align}
\left\lvert \mathbb{E}_h[E_{h,p}]\tr{\rho v_p} - \tilde{\mathbb{E}}[E_{h,p}\tr{\rho v_p}] \right\rvert &\leq 8 \delta \max_{h} \lvert E_{h,p} \tr{\rho v_p} \rvert \nonumber \\
  &\leq 8 \delta \max_{h,p} \lvert E_{h,p} \rvert,
\end{align}
where we used in the last step the unitarity of $v_k$ and the Von-Neumann trace inequality.
For an final error of $\epsilon$, we hence choose $\delta_t = \epsilon/[16\max_{h,p}|E_{h,p}|]$ to ensure that this sampling error incurrs at most
half the error budget of $\epsilon$.
Thus we ensure $\delta \le 1/4$ if $\epsilon \le 4 \max_{h,p} |E_{h,p}|$.

We can improve the query complexity of estimating the above expectation by values by using amplitude amplification, sice we obtain the measurement via a Hadamard test.
For this case we require only $O(\max_{h,p}|E_{h,p}|/\epsilon)$ samples in order to achieve the desired accuracy from the sampling.
Noting that we might not be able to even access $\tilde{H}_h$ without any error, we can deduce that the error of the individual terms of $\tilde{H}_h$ for an $\epsilon$-error in the final estimate must be bounded by $\delta_tv\norm{\theta}_1$,
where with abuse of notation, $\delta_t$ now denotes the error in the estimates of $E_{h,k}$.
Even taking this into account, the evaluation of this contribution is however dominated by the second term, and hence can be neglected in the analysis.

\end{proof}

\subsubsection{Approach 2: Training With Higher Order Divided Differences And Function Approximations}

In this section we develop a scheme to train a quantum Boltzmann machine using divided difference estimates for the relative entropy error.
generate differentiation formulas by differentiating an interpolant.
The idea for this is straightforward: First we construct an interpolating polynomial from the data.
Second, an approximation of the derivative at tany point can be then obtained by a direct differentiation of the interpolant.
We assume in the following that we can simulate and evaluate $\tr{\rho \log \sigma_v}$.
As this is generally non-trivial, and the error is typically large, we propose in the next section a different more specialised approach which, however, still allows us to train arbitrary models with the relative entropy objective.

In order to proof the error of the gradient estimation via interpolation, we first need to establish error bounds on the interpolating polynomial which can be obtained via the remainder of the Lagrange interpolation polynomial.
The gradient error for our objective can then be obtained by as a combination of this error with a bound on the $n+1$-st order derivative of the objective.
We start by bounding the error in the polynomial approximation.
\begin{lemma}
  \label{lem:remainder}
  Let $f(\theta)$ be the $n+1$ times differentiable function for which we want to approximate the gradient and let $p_n(\theta)$ be the degree $n$ Lagrange interpolation polynomial for points $\{\theta_1, \theta_2, \ldots, \theta_k, \ldots, \theta_n\}$.
  The gradient evaluated at point $\theta_k$ is then given by the interpolation polynomial
  \begin{equation}
    \frac{\partial p(\theta_k)}{\partial \theta} = \sum_{j=0}^n f(\theta_j) \mathcal{L}_{n,j}'(\theta_k),
  \end{equation}
  where $\mathcal{L}_{n,j}'$ is the derivative of the Lagrange interpolation polynomials $\mathcal{L}_{\mu,j}(\theta):= \prod_{\substack{k=0\\ k\neq j}}^{\mu} \frac{\theta - \theta_k}{\theta_j - \theta_k}$, and the error is given by
  \begin{equation}
    \left\lvert \frac{\partial f(\theta_k)}{\partial \theta} - \frac{\partial p_n(\theta_k)}{\partial \theta} \right\rvert \leq  \frac{1}{(n+1)!} \left\lvert f^{(n+1)}(\xi(\theta_k)) \prod\limits_{\substack{j=0 \\ j\neq k}}^n (\theta_j - \theta_k) \right\rvert,
  \end{equation}
  where $\xi(\theta_k)$ is a constant depending on the point $\theta_k$ at which we evaluated the gradient, and $f^{(i)}$ denotes the $i$-th derivative of $f$.
\end{lemma}
Note that $\theta$ is a point within the set of points at which we evaluate.
\begin{proof}
  Recall that the error for the degree $n$ Lagrange interpolation polynomial is given by
  \begin{equation}
    f(\theta) -  p_n(\theta) \leq  \frac{1}{(n+1)!} f^{(n+1)}(\xi_{\theta}) w(\theta),
  \end{equation}
  where $w(\theta) := \prod\limits_{j=1}^n(\theta - \theta_j)$.
  We want to estimate the gradient of this, and hence need to evaluate
  \begin{equation}
    \gradtheta{f(\theta)} -  \gradtheta{p_n(\theta)} \leq \lim\limits_{\Delta \rightarrow 0} \left( \frac{\frac{1}{(n+1)!} f^{(n+1)}(\xi_{\theta + \Delta}) w(\theta + \Delta) - \frac{1}{(n+1)!} f^{(n+1)}(\xi_{\theta}) w(\theta)}{\Delta} \right).
  \end{equation}
  Now, since we do not necessarily want to estimate the gradient at an arbitrary point $\theta$ but indeed have the freedom to choose the point, we can set $\theta$ to be one of the points at which we evaluate the function $f(\theta)$, i.e., $\theta \in \{ \theta_i \}_{i=1}^n$.
  Let this choice be given by $\theta_k$, arbitrarily chosen. Then we see that the latter term vanishes since $w(\theta_k)=0$. Therefore we have
  \begin{equation}
    \gradtheta{f(\theta_k)} -  \gradtheta{p_n(\theta_k)} \leq \lim\limits_{\Delta \rightarrow 0} \left( \frac{\frac{1}{(n+1)!} f^{(n+1)}(\xi_{\theta_k + \Delta}) w(\theta_k + \Delta)}{\Delta} \right),
  \end{equation}
  and noting that $w(\theta_k)$ contains one term $(\theta_k + \Delta - \theta_k) = \Delta$ achieves the claimed result.
\end{proof}

We will perform a number of approximation steps in order to obtain a form which can be simulated on a quantum computer more efficiently, and only then resolve to divided differences at this ``lower level".
In detail we will perform the following steps.
As described in the body of the paper, we perform the following steps in order to obtain the gradient.
\begin{enumerate}
  \item Approximate the logarithm via a Fourier-like approximation
  \begin{equation}
    \log \sigma_v \rightarrow \log_{K,M}\sigma_v,
  \end{equation}
  which yields a Fourier-like series $\sum_m c_m \exp{(im\pi \sigma_v)}$.
  \item Evaluate the gradient of $\tr{ \frac{\partial}{\partial \theta} \rho \log_{K,M}(\sigma_v)}$,
  yielding terms of the form
  \begin{equation}
    \int_0^1 ds e^{(ism\pi \sigma_v)} \frac{\partial \sigma_v}{\partial \theta} e^{(i(1-s)m\pi \sigma_v)}.
  \end{equation}
  \item Each term in this expansion can be evaluated separately via a sampling procedure, i.e.,
  \begin{equation}
    \int_0^1 ds e^{(ism\pi \sigma_v)} \frac{\partial \sigma_v}{\partial \theta} e^{(i(1-s) m\pi \sigma_v)} \approx \mathbb{E}_s
    \left[ e^{(ism\pi \sigma_v)} \frac{\partial \sigma_v}{\partial \theta} e^{(i(1-s)m\pi \sigma_v)} \right].
  \end{equation}
  \item Apply a divided difference scheme to approximate the gradient $\frac{\partial \sigma_v}{\partial \theta}$.
  \item Use the Fourier series approach to aproximate the density operator $\sigma_v$ by the series of itself,
  i.e., $\sigma_v \approx F(\sigma_v) := \sum_{m'} c_{m'} \exp{(im \pi m' \sigma_v)}$.
  \item Evaluate these terms conveniently via sample based Hamiltonian simulation and the Hadamard test.
\end{enumerate}
In the following we will give concrete bounds on the error introduced by the approximations and details of the implementation.
The final result is then stated in Theorem~\ref{thm:complexity_general_algo}.
We first bound the error in the approximation of the logarithm and then use Lemma~37 of \cite{van2017quantum} to obtain a Fouries series approximation which is close to $\log(z)$.
  The Taylor series of $$\log(x) = \sum_{k=1}^{\infty} (-1)^{k+1}\frac{(x-1)^k}{k} = \sum_{k=1}^{K_1} (-1)^{k+1}\frac{(x-1)^k}{k} + R_{K_1+1}(x-1),$$ for $x \in (0,1)$ and where $R_{K+1}(z)= \frac{f^{K_1+1}(c)}{K!}(z-c)^{K_1}z$ is the Cauchy remainder of the Taylor series, for $-1<z<0$.
  The error can hence be bounded as $$\lvert R_{K_1+1}(z) \rvert = \left\lvert (-1)^{K_1}\frac{z^{K_1+1}(1-\alpha)^{K_1}}{(1+\alpha z)^{K_1+1}}\right\rvert,$$
  where we evaluated the derivatives of the logarithm and $0\leq \alpha \leq 1$ is a parameter. Using that $1+\alpha z \geq 1+z$ (since $z \leq 0$) and hence $0 \leq \frac{1-\alpha}{1+\alpha z} \leq 1$, we obtain the error bound
  \begin{equation}
    \lvert R_{K_1+1}(z) \rvert \leq  \frac{\left\lvert z \right\rvert^{K_1+1}}{1+z}
  \end{equation}
  Reversing to the variable $x$ the error bound for the Taylor series, and assuming that $0<\delta_l <z$ and  $0<|1-z|\leq \delta_u <1$, which is justified if we are dealing with sufficiently mixed states, then the approximation error is given by
  \begin{equation}
    \lvert R_{K_1+1}(z) \rvert \leq  \frac{(\delta_l) ^{K_1+1}}{\delta_u} \overset{!}{\leq} \epsilon_1.
  \end{equation}
  Hence in order to achieve the desired error $\epsilon_1$ we need $$K_1 \geq \frac{\log \left( (\epsilon_1 \delta_u)^{-1} \right)}{\log \left((\delta_l)^{-1}\right)}.$$
  We hence can chose $K_1$ such that the error in the approximation of the Taylor series is $\leq \epsilon_1/4$.
  This implies we can make use of Lemma~37 of \cite{van2017quantum}, and therefore obtain a Fourier series approximation for the logarithm.
  We will restate this Lemma here for completeness:

  \begin{lemma}[Lemma 37, \cite{van2017quantum}]
    \label{lem:vanApeldoorn}
    Let $f:\mathbb{R} \rightarrow \mathbb{C}$ and $\delta,\epsilon \in (0,1)$, and $T(f):= \sum_{k=0}^K a_k x^k$ be a polynomial such that $\left\lvert f(x) - T(f) \right\rvert \leq \epsilon/4$ for all $x \in [-1+\delta, 1-\delta]$.
    Then $\exists c \in \mathbb{C}^{2M+1}:$
    \begin{equation}
      \left\lvert f(x) - \sum_{m=-M}^{M} c_m e^{\frac{i \pi m}{2}x} \right\rvert \leq \epsilon
    \end{equation}
    for all $x \in [-1+\delta, 1- \delta]$, where $M= \max \left(2 \left\lceil \ln \left( \frac{4 \norm{a}_1}{\epsilon} \right) \frac{1}{\delta} \right\rceil, 0\right)$ and $\norm{c}_1 \leq \norm{a}_1$.
    Moreover, $c$ can be efficiently calculated on a classical computer in time $\mathrm{poly}(K,M,\log(1/\epsilon))$.
  \end{lemma}
  In order to apply this lemma to our case, we restrict the approximation rate to the range $(\delta_l, \delta_u)$, where $0 < \delta_l \leq \delta_u <1$.
  Therefore we obtain over this range a approximation of the following form.
  \begin{corollary}
    Let $f:\mathbb{R} \rightarrow \mathbb{C}$ be defined as $f(x)=\log(x)$, $\delta,\epsilon_1 \in (0,1)$, and $\log_{K}(1-x):= \sum_{k=1}^{K_1} \frac{(-1)^{k-1}}{k} x^k$ such that $a_k:=\frac{(-1)^{k-1}}{k}$ and $\norm{a}_1 = \sum_{k=1}^{K_1} \frac{1}{k}$ with $K_1 \geq \frac{\log \left(4(\epsilon_1 \delta_1^u)^{-1} \right)}{\log \left((\delta_l)^{-1}\right)}$ such that $\left\lvert \log(x) -\log_{K}(x)\right\rvert \leq \epsilon_1/4$ for all $x \in [\delta_l, \delta_u]$.
    Then $\exists c \in \mathbb{C}^{2M+1}:$
    \begin{equation}
      \left\lvert f(x) - \sum_{m=-M_1}^{M_1} c_m e^{\frac{i \pi m}{2}x} \right\rvert \leq \epsilon_1
    \end{equation}
    for all $x \in [\delta_l, \delta_u]$, where $M_1= \max \left(2 \left\lceil \ln \left( \frac{4 \norm{a}_1}{\epsilon_1} \right) \frac{1}{1-\delta_u} \right\rceil, 0\right)$ and $\norm{c}_1 \leq \norm{a}_1$.
    Moreover, $c$ can be efficiently calculated on a classical computer in time $\mathrm{poly}(K_1,M_1,\log(1/\epsilon_1))$.
  \end{corollary}
  \begin{proof}
    The proof follows straight forward by combining Lemma~\ref{lem:vanApeldoorn} with the approximation of the logarithm and the range over which we want to approximate the function.
  \end{proof}
  In the following we denote with $\log_{K,M}(x):=\sum_{m=-M_1}^{M_1} c_m e^{\frac{i \pi m}{2}x}$, where we keep the $K$-subscript to denote that classical computation of this approximation is $\mathrm{poly}(K)$-dependent.
  We can now express the gradient of the objective via this approximation as
  \begin{equation}
    \label{eq:exact_gradient_log_approx}
    \tr{\frac{\partial}{\partial \theta} \rho \log_{K,M} \sigma_v} \approx \sum_{m=-M_1}^{M_1} \frac{i c_m m \pi}{2} \int_0^1 ds\ \tr{ \rho e^{\frac{i s \pi m}{2}\sigma_v} \frac{\partial \sigma_v}{\partial \theta} e^{\frac{i (1-s) \pi m}{2}\sigma_v}}.
  \end{equation}
  where we can evaluate each term in the sum individually and then classically post process the results, i.e., sum these up. In particular the latter can be evaluated as the expectation value over $s$, i.e.,
  \begin{equation}
    \int_0^1 ds\ \tr{ \rho e^{\frac{i s \pi m}{2}\sigma_v} \frac{\partial \sigma_v}{\partial \theta} e^{\frac{i (1-s) \pi m}{2}\sigma_v}} = \mathbb{E}_{s \in[0,1]} \left[\tr{ \rho e^{\frac{i s \pi m}{2}\sigma_v} \frac{\partial \sigma_v}{\partial \theta} e^{\frac{i (1-s) \pi m}{2}\sigma_v}} \right],
  \end{equation}
 which we can evaluate separately on a quantum device.
 In the following we hence need to device a method to evaluate this expectation value.\\
 First, we will expand the gradient using a divided difference formula such that $\frac{\partial \sigma_v}{\partial \theta}$ is approximated by the Lagrange interpolation polynomial of degree $\mu-1$, i.e., $$\frac{\partial \sigma_v}{\partial \theta}(\theta) \approx \sum_{j=0}^{\mu}\sigma_v(\theta_j) \mathcal{L'}_{\mu,j}(\theta),$$ where $$\mathcal{L}_{\mu,j}(\theta):= \prod_{\substack{k=0\\ k\neq j}}^{\mu} \frac{\theta - \theta_k}{\theta_j - \theta_k}.$$
 Note that the order $\mu$ is free to chose, and will guarantee a different error in the solution of the gradient estimate as described prior in Lemma~\ref{lem:remainder}. Using this in the gradient estimation, we obtain a polynomial of the form (evaluated at $\theta_j$, i.e., the chosen points)
\begin{equation}
 \sum_{m=-M_1}^{M_1} \frac{i c_m m \pi}{2} \sum_{j=0}^{\mu} \mathcal{L'}_{\mu,j}(\theta_j ) \mathbb{E}_{s \in[0,1]} \left[\tr{ \rho e^{\frac{i s \pi m}{2}\sigma_v} \sigma_v(\theta_j) e^{\frac{i (1-s) \pi m}{2}\sigma_v}} \right],
\end{equation}
where each term again can be evaluated separately, and efficiently combined via classical post processing. Note that the error in the Lagrange interpolation polynomial decreases exponentially fast, and therefore the number of terms we use is sufficiently small to do so.
Next, we need to deploy a method to evaluate the above expressions.
In order to do so, we implement $\sigma_v$ as a Fourier series of itself, i.e., $\sigma_v = \arcsin(\sin(\sigma_v \pi/2)/(\pi/2))$, which we will then approximate similar to the approach taken in Lemma~\ref{lem:vanApeldoorn}. With this we obtain the following result.

\begin{lemma}\label{lem:bounds_z_approx}
  Let $\delta,\epsilon_2 \in (0,1)$, and $\tilde x:= \sum_{m'=-M_2}^{M_2} \tilde{c}_{m'} e^{i \pi m'x/2}$ with  $K_2 \geq \frac{\log(4/\epsilon_2)}{\log(\delta_u^{-1})}$ and $M_2 \geq \left\lceil \log \left(\frac{4}{\epsilon_2} \right) \sqrt{(2\log{\delta_u^{-1}})^{-1}} \right\rceil$ and $x \in [\delta_l, \delta_u]$.
  Then $\exists \tilde c \in \mathbb{C}^{2M+1}:$
  \begin{equation}
    \left\lvert x - \tilde x \right\rvert \leq \epsilon_2
  \end{equation}
  for all $x \in [\delta_l, \delta_u]$, and $\norm{c}_1 \leq 1$.
  Moreover, $\tilde c$ can be efficiently calculated on a classical computer in time $\mathrm{poly}(K_2,M_2,\log(1/\epsilon_2))$.
\end{lemma}

\begin{proof}
Invoking the technique used in~\cite{van2017quantum}, we expand $$\arcsin(z) = \sum_{k'=0}^{K_2} 2^{-2k'} {2k' \choose k'} \frac{z^{2k'+1}}{2k'+1} + R_{K_2+1}(z),$$ wher $R_{K_2+1}$ is the remainder as before.
For $0<z\leq \delta_u \leq 1/2$, remainder can be bound by $ \left\lvert R_{K_2+1} \right\rvert \leq \frac{\left\lvert \delta_u \right\rvert^{K_2+1}}{1/2} \overset{!}{\leq} \epsilon_2/2$, which gives the bound $$K_2 \geq \frac{\log(4/\epsilon_2)}{\log(\delta_u^{-1})}.$$
We then approximate
\begin{equation}
\sin^l(x)=\left(\frac{i}{2}\right)^l \sum_{m'=0}^l (-1)^{m'} {l \choose m'} e^{ix(2m'-l)}
\end{equation} by
\begin{equation}
\sin^l(x)\approx \left(\frac{i}{2}\right)^l \sum_{m'=\lceil l/2\rceil - M_2}^{\lfloor l/2\rfloor + M_2} (-1)^{m'} {l \choose m'} e^{ix(2m'-l)},
\end{equation}
which induces an error of $\epsilon_2/2$ for the choice $$M_2 \geq \left\lceil \log \left(\frac{4}{\epsilon_2} \right) \sqrt{(2\log{\delta_u^{-1}})^{-1}} \right\rceil.$$
This can be seen by using Chernoff's inequality for sums of binomial coefficients, i.e., $$\sum_{m'=\lceil l/2+M_2\rceil}^l 2^{-l} {l \choose m'} \leq e^{-\frac{2M_2^2}{l}},$$ and chosing $M$ appropriately.
Finally, defining $f(z):= \arcsin(\sin(z \pi/2)/(\pi/2))$, as well as $\tilde{f}_1 := \sum_{k'=0}^{K_2} b_{k'} \sin^{2k'+1}(z \pi/2)$ and
\begin{equation}
\tilde{f}_2(z) := \sum_{k'=0}^{K_2} b_{k'}\left(\frac{i}{2}\right)^l \sum_{m'=\lceil l/2\rceil - M_2}^{\lfloor l/2\rfloor + M_2} (-1)^{m'} {l \choose m'} e^{ix(2m'-l)},
\end{equation}
 and observing that
 $$\norm{f -\tilde{f}_2}_{\infty} \leq \norm{f - \tilde{f}_1}_{\infty} + \norm{\tilde{f}_1 - \tilde{f}_2}_{\infty},$$ yields the final error of $\epsilon_2$ for the approximation
$z \approx \tilde z = \sum_{m'} \tilde{c}_{m'} e^{i \pi m' z/2}$.
\end{proof}
Note that this immediately leads to an $\epsilon_2$ error in the spectral norm for the approximation
\begin{equation}
\left\lVert \sigma_v-\sum_{m'=-M_2}^{M_2} \tilde{c}_{m'} e^{i \pi m' \sigma_v/2}\right\rVert_2 \leq \epsilon_2,
\end{equation}
where $\sigma_v$ is the reduced density matrix.

Since our final goal is to estimate $\tr{\partial_{\theta} \rho \log \sigma_v}$, with a variety of $\sigma_v(\theta_j)$ using the divided difference approach, we also need to bound the error in this estimate which we introduce with the above approximations.
Bounding the derivative with respect to the remainder can be done by using the truncated series expansion and bounding the gradient of the remainder. This yields the following result.

\begin{lemma}
    \label{lem:proof_error_bound_grads}
    For the of the parameters $M_1,M_2, K_1,L,\mu,\Delta,s$ given in eq.~(\ref{eq:bound_M1}-\ref{eq:bounds_epss}), and $\rho, \sigma_v$ being two density matrices, we can estimate the gradient of the relative entropy such that
    \begin{equation}
      \left\lvert \partial_{\theta}\tr{\rho \log \sigma_v} - \partial_{\theta} \tr{\rho \log_{K_1,M_1} \tilde{\sigma}_v} \right\rvert \leq \epsilon,
    \end{equation}
    where the function $\partial_{\theta} \tr{\rho \log_{K_1,M_1} \tilde{\sigma}_v}$ evaluated at $\theta$ is defined as
    \begin{equation}
      \label{eq:full_approximation}
    \mathrm{Re}\left[ \sum_{m=-M_1}^{M_1} \sum_{m'=-M_2}^{M_2} \frac{i c_m \tilde{c}_{m'} m \pi}{2} \sum_{j=0}^{\mu} \mathcal{L}'_{\mu,j}(\theta) \mathbb{E}_{s \in[0,1]} \left[\tr{ \rho e^{\frac{i s \pi m}{2}\sigma_v} e^{\frac{i \pi m'}{2} \sigma_v(\theta_j)} e^{\frac{i (1-s) \pi m}{2}\sigma_v}} \right] \right]
    \end{equation}
    The gradient can hence be approximated to error $\epsilon$ with $O(\text{poly}(M_1,M_2,K_1,L,s,\Delta,\mu))$ computation on a classical computer and using only the Hadamard test, Gibbs state preparation and LCU on a quantum device.
\end{lemma}
Notably the expression in \eqref{eq:full_approximation} can now be evaluated with a quantum-classical hybrid device by evaluating each term in the trace separately via a Hadamard test and, since the number of terms is only polynomial, and then evaluating the whole sum efficiently on a classical device.\\

\begin{proof}
  For the proof we perform the following steps.
  Let $\sigma_i(\rho)$ be the singular values of $\rho$, which are equivalently the eigenvalues since $\rho$ is Hermitian.
  Then observe that the gradient can be separated in different terms, i.e., let $\log_{K_1,M_1}^s\sigma_v$ be the approximation as given in \eqref{eq:full_approximation} for a finite sample of the expectation values $\mathbb{E}_{s}$, then we have
  \begin{align}
	\label{eq:bounds_approximation_error}
    &\left\lvert \partial_{\theta}\tr{\rho \log \sigma_v} - \partial_{\theta} \tr{\rho \log_{K_1,M_1}^s \tilde{\sigma}_v} \right\rvert \leq \nonumber \\
      &\leq \sum_i \sigma_i(\rho) \cdot \left\lVert \partial_{\theta} [\log \sigma_v - \log_{K_1,M_1}^s \tilde{\sigma}_v] \right\rVert \nonumber \\
      &\leq \sum_i \sigma_i(\rho) \cdot \left( \left\lVert \partial_{\theta} [\log \sigma_v - \log_{K_1,M_1} \sigma_v] \right\rVert \right. \nonumber \\
      &+ \left. \left\lVert \partial_{\theta} [\log_{K_1,M_1} \sigma_v - \log_{K_1,M_1} \tilde{\sigma}_v] \right\rVert + \left\lVert \partial_{\theta} [\log_{K_1,M_1} \tilde{\sigma_v} - \log^s_{K_1,M_1} \tilde{\sigma}_v] \right\rVert \right)
  \end{align}
  where the second step follows from the Von-Neumann trace inequality and the terms are (1) the error in approximating the logarithm, (2) the error introduced by the divided difference and the approximation of $\sigma_v$ as a Fourier-like series, and (3) is the finite sampling approximation error.
  We can now bound the different term separately, and start with the first part which is in general harder to estimate.
  We partition the bound in three terms, corresponding to the three different approximations taken above.
  \begin{align*}
      &\left\lVert \partial_{\theta} [\log \sigma_v - \log_{K_1,M_1} \sigma_v] \right\rVert \leq \\
      &\leq \left\lVert \partial_{\theta} \sum_{k=K_1+1}^{\infty} \frac{(-1)^k}{k} \sigma_v^k \right\rVert + \left\lVert \partial_{\theta} \sum_{k=1}^{K_1} \frac{(-1)^k}{k} \sum_{l=L}^{\infty} b_l^{(k)} \sin^l(\sigma_v \pi/2) \right\rVert \\
      & +\left\lVert \partial_{\theta} \sum_{k=1}^{K_1} \frac{(-1)^k}{k} \sum_{l=L}^{\infty} b_l^{(k)}  \left(\frac{i}{2} \right)^l \sum_{m \in [0, \lceil l/2\rceil -M_1] \cup [\lfloor l/2 \rfloor +M_1, l]} (-1)^m e^{i (2m-l)\sigma_v \pi/2} \right\rVert
  \end{align*}
  The first term can be bound in the following way:
  \begin{align}
    \leq \sum_{k=K_1+1}^{\infty} \lVert \sigma_v \rVert^{k-1} = \frac{\lVert \sigma_v\rVert^{K_1}}{1-\lVert \sigma_v \rVert},
  \end{align}
  and, assuming $\norm{\sigma_v}<1$, we hence can set
\begin{equation}
	K_1\geq \log((1-\norm{\sigma_v})\epsilon/9)/\log(\norm{\sigma_v})
\end{equation}
 appropriately in order to achieve an $\epsilon/9$ error. The second term can be bound by assuming that $\norm{\sigma_v \pi}<1$, and chosing $$L\geq\log \frac{\left(\frac{\epsilon}{9\pi K \norm{\frac{\partial \sigma_v}{\partial\theta}}}\right)}{\log(\norm{\sigma_v} \pi)},$$ which we derive by observing that
  \begin{align}
    &\leq \sum_{k=1}^K \frac{1}{k} \sum_{l=L}^{\infty} b_l^{(k)} l \left\lVert \sin^{l-1} (\sigma_v \pi/2)\right\rVert \cdot \left\lVert \frac{\pi}{2} \frac{\partial \sigma_v}{\partial \theta} \right\rVert \\
    & < \sum_{k=1}^K \frac{1}{k} \sum_{l=L+1}^{\infty} b_l^{(k)} \pi \left\lVert \sigma_v \pi\right\rVert^{l-1} \cdot \left\lVert \frac{\partial \sigma_v}{\partial \theta} \right\rVert \\
    &\leq \sum_{k=1}^K \frac{1}{k}  \pi \left\lVert \sigma_v \pi\right\rVert^{L} \cdot \left\lVert \frac{\partial \sigma_v}{\partial \theta} \right\rVert,
  \end{align}
where we used in the second step that $l < 2^l$.
  Finally, the last term can be bound similarly, which yields
  \begin{align}
   &\leq \sum_{k=1}^K \frac{1}{k} \sum_{l=1}^{L} b_l^{(k)} e^{-2(M_1)^2/l} \cdot l \cdot \frac{\pi}{2} \norm{\frac{\partial \sigma_v}{\partial \theta}}\\
   &\leq \sum_{k=1}^K \frac{L}{k} \sum_{l=1}^{L} b_l^{(k)} e^{-2(M_1)^2/L} \frac{\pi}{2} \norm{\frac{\partial \sigma_v}{\partial \theta}} \\
   &\leq \sum_{k=1}^K \frac{L}{k} e^{-2(M_1)^2/L} \frac{\pi}{2} \norm{\frac{\partial \sigma_v}{\partial \theta}} \leq \frac{KL\pi}{2}e^{-2(M_1)^2/L} \norm{\frac{\partial \sigma_v}{\partial \theta}},
  \end{align}
  and we can hence chose $$M_1\geq \sqrt{L \log\left(\frac{9 \norm{\frac{\partial \sigma_v}{\partial \theta}} K_1 L \pi}{2\epsilon}\right)} $$ in order to decrease the error to $\epsilon/3$ for the first term in \eqref{eq:bounds_approximation_error}.\\
  For the second term, first note that with the notation we chose, $\left\lVert \partial_{\theta} [\log_{K_1,M_1} \sigma_v - \log_{K_1,M_1} \tilde{\sigma}_v] \right\rVert$ is the difference between the log-approximation where the gradient of $\sigma_v$ is still exact, i.e., \eqref{eq:exact_gradient_log_approx}, and the version where we approximate the gradient via divided differences and the linear combination of unitaries, given in \eqref{eq:full_approximation}.
  Recall that the first level approximation was given by
  \begin{equation*}
    \sum_{m=-M_1}^{M_1} \frac{i c_m m \pi}{2} \int_0^1 ds\ \tr{ \rho e^{\frac{i s \pi m}{2}\sigma_v} \frac{\partial \sigma_v}{\partial \theta} e^{\frac{i (1-s) \pi m}{2}\sigma_v}},
  \end{equation*}
  where we went from the expectation value formulation back to the integral formulation to avoid consideration of potential errors due to sampling.

  Bounding the difference hence yields one term from the divided difference approximation of the gradient and an error from the Fourier series, which we can both bound separately. Denoting $\partial \tilde p(\theta_k)/\partial \theta$ as the divided difference and the LCU approximation of the Fourier series\footnote{which effectively means that we approximate the coefficients of the interpolation polynomial}, and with $\partial p(\theta_k)/\partial \theta$ the divided difference without approximation via the Fourier series,  we hence have
  \begin{align}
    &\left\lVert \partial_{\theta} [\log_{K_1,M_1} \sigma_v - \log_{K_1,M_1} \tilde{\sigma}_v] \right\rVert \leq \\
    &\leq \left\lvert \sum_{m=-M_1}^{M_1} \frac{i c_m m \pi}{2} \int_0^1 ds\ \tr{ \rho e^{\frac{i s \pi m}{2}\sigma_v} \left( \frac{\partial \sigma_v}{\partial \theta} - \frac{\partial \tilde p(\theta_k)}{\partial \theta} \right) e^{\frac{i (1-s) \pi m}{2}\sigma_v}} \right\rvert \\
    &\leq \frac{M_1 \pi \norm{a}_1}{2}  \int_0^1 ds\ \sum_i \sigma_i(\rho) \left\lVert \frac{\partial \sigma_v}{\partial \theta} -  \frac{\partial \tilde p(\theta_k)}{\partial \theta}\right\rVert \\
&\leq \frac{M_1 \pi \norm{a}_1}{2}  \int_0^1 ds\ \sum_i \sigma_i(\rho) \left( \left\lVert \frac{\partial \sigma_v}{\partial \theta} - \frac{\partial p(\theta_k)}{\partial \theta}\right\rVert  +  \left\lVert \frac{\partial p(\theta_k)}{\partial \theta} -  \frac{\partial \tilde p(\theta_k)}{\partial \theta}\right\rVert \right) \\
&\leq \frac{M_1 \norm{a}_1 \pi}{2} \int_0^1 ds\ \sum_i \sigma_i(\rho) \left(\norm{\frac{\partial^{\mu+1} \sigma_v}{\partial \theta^{\mu+1}}} \left( \frac{\Delta}{\mu-1}\right)^{\mu}\frac{ \max_k (\mu-k)!}{(\mu+1)!} + \sum_{j=0}^{\mu} \lvert \mathcal{L}'_{\mu,j}(\theta_j) \rvert \norm{\sigma_v - \tilde{\sigma}_v} \right) \\
    &\leq \frac{M_1 \norm{a}_1 \pi}{2} \left( \norm{\frac{\partial^{\mu+1} \sigma_v}{\partial \theta^{\mu+1}}} \left( \frac{\Delta}{\mu-1}\right)^{\mu} \frac{\mu!}{(\mu+1)!} +\mu \norm{\mathcal{L}'_{\mu,j}(\theta_j)}_{\infty} \epsilon_2 \right) , \label{eq:bounds_ddfs}
  \end{align}
  where $\norm{a}_1 = \sum_{k=1}^{K_1} 1/k$, and we used in the last step the results of Lemma~\ref{lem:bounds_z_approx}. Under appropriate assumptions on the grid-spacing for the divided difference scheme $\Delta$ and the number of evaluated points $\mu$ as well as a bound on the $\mu+1$-st derivative of $\sigma_v$ w.r.t. $\theta$, we can hence also bound this error.
In order to do so, we need to analyze the $\mu+1$-st derivative of $\sigma_v = \trh{e^{-H}}/Z$ with $Z=\tr{e^{-H}}$.
For this we have
  \begin{align}
\label{eq:bound_mu_derivative}
  \norm{\frac{\partial^{\mu+1} \sigma_v}{\partial \theta^{\mu+1}}} &\leq \sum_{p=1}^{\mu+1} {\mu+1 \choose p}  \norm{\frac{\partial^{p} \trh{e^{-H}}}{\partial \theta^{p}}}  \norm{\frac{\partial^{\mu+1-p} Z^{-1}}{\partial \theta^{\mu+1-p}}} \nonumber\\
&\leq  2^{\mu+1} \max_p \norm{\frac{\partial^{p} \trh{e^{-H}}}{\partial \theta^{p}}}  \norm{\frac{\partial^{\mu+1-p} Z^{-1}}{\partial \theta^{\mu+1-p}}}
  \end{align}
We have that
\begin{align}
\norm{\frac{\partial^{p} \trh{e^{-H}}}{\partial \theta^{p}}} &\leq  \mathrm{dim}(H_h) \norm{\frac{\partial^q e^{-H}}{\partial \theta^q}}
\end{align}
where $\mathrm{dim}(H_h) = 2^{n_h}$. In order to bound this, we take advantage of the infinitesimal expansion of the exponent, i.e.,
\begin{align}
	\norm{\frac{\partial^q e^{-H}}{\partial \theta^q}} &= \norm{\frac{\partial^q }{\partial \theta^q} \lim_{r\rightarrow \infty} \prod_{j=1}^{r} e^{-H/r}} \nonumber \\
	&= \norm{\lim_{r \rightarrow \infty} \left(  \frac{\partial^q e^{-H/r}}{\partial \theta^q} \prod_{j=2}^{r} e^{-H/r} +  \frac{\partial^{q-1} e^{-H/r}}{\partial \theta^{q-1}} \frac{\partial e^{-H/r}}{\partial \theta} \prod_{j=3}^{r} e^{-H/r} + \ldots \right)} \nonumber \\
	&\leq \lim_{r \rightarrow \infty} \left(\norm{\frac{\partial H/r}{\partial \theta}}^q \cdot r^q + O\left(\frac{1}{r}\right) \right) \norm{e^{-H}} = \norm{\frac{\partial H}{\partial \theta}}^q\norm{e^{-H}}  ,
\end{align}
where the last step follows from the fact that we have $r^q$ terms and that we used that the error introduced by the commutations above will be of $O(1/r)$.
Observing that $\partial_{\theta_i}H = \partial_{\theta_i}\sum_j \theta_j H_j = H_i$ and assuming that $\lambda_{max}$ is the largest singular eigenvalue of $H$, we can hence bound this by $\lambda_{max}^q\norm{e^{-H}}$.

\begin{align}
\norm{\frac{\partial^{p} \trh{e^{-H}}}{\partial \theta^{p}}} &\leq  \mathrm{dim}(H_h) \norm{\frac{\partial^q e^{-H}}{\partial \theta^q}} \nonumber \\
	&\leq  \lambda_{max}^{p}  \mathrm{dim}(H_h) \norm{ \trh{e^{-H}}},
\end{align}

\begin{align}
\norm{\frac{\partial^{\mu+1-p} Z^{-1}}{\partial \theta^{\mu+1-p}}} &\leq  \frac{(\mu+1-p)! \lvert \lambda_{max} \rvert^{\mu+1-p}}{Z^{\mu+2-p}} \tr{e^{-H}}\nonumber\\
	&\leq \left(\frac{\mu+1-p}{eZ}\right)^{\mu+1-p} \frac{e}{Z} \lvert \lambda_{max} \rvert^{\mu+1-p}\tr{e^{-H}} \nonumber\\
	&= \left(\frac{\mu+1-p}{eZ}\right)^{\mu+1-p} e \lvert \lambda_{max} \rvert^{\mu+1-p}
\end{align}
We can therefore find a bound for \eqref{eq:bound_mu_derivative} as
\begin{align}
      \norm{\frac{\partial^{\mu+1} \sigma_v}{\partial \theta^{\mu+1}}} &\leq  e 2^{\mu+1+n_h} \lambda_{max}^{\mu+1} \norm{\trh{e^{-H}}}  \max_p \left(\frac{\mu+1-p}{eZ}\right)^{\mu+1-p}.
\end{align}
Plugging this result into the bound from above yields
  \begin{align}
    &\left\lVert \partial_{\theta} [\log_{K_1,M_1} \sigma_v - \log_{K_1,M_1} \tilde{\sigma}_v] \right\rVert  \nonumber \\
    &\leq \frac{M_1 \norm{a}_1 \pi}{2} \left(e2^{\mu+1+n_h} \lambda_{max}^{\mu+1} \norm{\trh{e^{-H}}}  \max_p \left(\frac{\mu+1-p}{eZ}\right)^{\mu+1-p} \left( \frac{\Delta}{\mu-1}\right)^{\mu} \frac{1}{\mu+1}  \right) \nonumber \\
 	&+  \frac{M_1 \norm{a}_1 \pi}{2} \left( \mu \norm{\mathcal{L}'_{\mu,j}(\theta_j)}_{\infty}  \epsilon_2 \right) ,
  \end{align}
Note that under the reasonable assumption that $2 \leq \mu \ll Z$, the maximum is achieved for $p=\mu+1$, and we hence obtain the upper bound
 \begin{align}
	\label{eq:error_ddfs}
  &\frac{M_1 \norm{a}_1 \pi}{2} \left( 2^{n_h} e ( 2 \lvert \lambda_{max} \rvert)^{\mu+1}  \norm{\trh{e^{-H}}} \left( \frac{\Delta}{\mu-1}\right)^{\mu} \frac{1}{\mu+1}  + \mu  \norm{\mathcal{L}'_{\mu,j}(\theta_j)}_{\infty}   \epsilon_2 \right) \nonumber \\
 &\leq  \frac{M_1 \norm{a}_1 \pi}{2} \left( 2^{n_h}e ( 2 \lvert \lambda_{max} \rvert)^{\mu+1}  \norm{\trh{e^{-H}}} \left( \frac{\Delta}{\mu-1}\right)^{\mu} + \mu  \norm{\mathcal{L}'_{\mu,j}(\theta_j)}_{\infty}  \epsilon_2 \right) ,
  \end{align}
and we can hence obtain a bound on $\mu$, the grid point number, in order to achieve an error of $\epsilon/6 > 0$ for the former term, which is given by
\begin{equation}
	\mu \geq (\lvert \lambda_{max} \rvert \Delta)  \exp \left( W \left(
			 \frac{\log \left( 2^{n_h}\frac{
			6  M_1 \norm{a}_1 e^2 \lvert \lambda_{max}\rvert \pi  \norm{\trh{e^{-H}}}}{\epsilon}
			\right)}{2 \lambda_{max} \Delta}
	\right)\right),
\end{equation}
where $W$ is the Lambert function, also known as product-log function, which generally grows slower than the logarithm in the asymptotic limit.
Note that $\mu$ can hence be lower bounded by

\begin{equation}
\label{eq:bound_mu}
	\mu \geq n_h + \log \left( \frac{6 M_1 \norm{a}_1 e^2 \lvert \lambda_{max}\rvert \pi  \norm{\trh{e^{-H}}}}{\epsilon} \right):= n_h + \log \left(\frac{M_1\Lambda}{\epsilon} \right).
\end{equation}
For convenience, let us choose $\epsilon$ such that $n_h + \log(M_1 \Lambda/\epsilon)$ is an integer.  We do this simply to avoid having to keep track of ceiling or floor functions in the following discussion where we will choose $\mu = n_h + \log(M_1\Lambda /\epsilon)$.

For the second part, we will bound the derivative of the Lagrangian interpolation polynomial.
First, note that $\mathcal{L}'_{\mu,j}(\theta)  = \sum_{l=0; l \neq j}^{\mu} \left( \prod_{k=0; k \neq j,l} \frac{\theta - \theta_k}{\theta_j - \theta_k} \right)\frac{1}{\theta_j-\theta_l}$ for a chosen discretization of the space such that $\theta_k -\theta_j= (k-j)\Delta/\mu$ can be bound by using a central difference formula, such that we use an uneven number of points (i.e. we take $\mu = 2\kappa +1$ for positive integer $\kappa$) and chose the point $m$ at which we evaluate the gradient as the central point of the mesh.
Note that in this case he have that for $\mu \ge 5$ and $\theta_m$ being the parameters at the midpoint of the stencil
\begin{align}
  \norm{\mathcal{L}_{\mu,j}'}_{\infty}
  &\leq \sum_{\substack{l\neq j}} \prod_{\substack{k \neq j,l}} \frac{|\theta_m - \theta_k|}{|\theta_j - \theta_k|} \frac{1}{\lvert \theta_l - \theta_j \rvert}
  \leq \frac{(\kappa!)^2}{(\kappa!)^2} \frac{\mu}{\Delta} \sum_{l \ne j} \frac{1}{|l -j|}\nonumber\\
  &\leq  \frac{2\mu}{\Delta} \sum_{l=1}^\kappa \frac{1}{l}\nonumber\\
  &\leq \frac{2\mu}{\Delta} \left( 1+ \int_{1}^{\kappa-1} \frac{1}{\ell} \mathrm{d}\ell\right) = \frac{2\mu}{\Delta}\left(1+\log((\mu-3)/2)\right)\le \frac{5\mu}{\Delta} \log(\mu/2),
\end{align}
where the last inequality follows from the fact that $\mu \ge 5$ and $1+\ln(5/2) < (5/2)\ln(5/2)$.
Now, plugging in the $\mu$ from \eqref{eq:bound_mu}, we find that this error is bound by
\begin{equation}
	\norm{\mathcal{L}_{\mu,j}'}_{\infty} \leq \frac{5 n_h + 5\log \left(\frac{M_1\Lambda}{\epsilon} \right)}{\Delta} \log(n_h/2 + \log \left(\frac{M_1\Lambda}{\epsilon} \right)/2) = \tilde{O}\left(\frac{n_h+\log\left(\frac{M_1\Lambda}{\epsilon} \right)}{\Delta}\right),
\end{equation}
If we want an upper bound of $\epsilon/6$ for the second term of the error in \eqref{eq:error_ddfs}, we hence require
%\begin{equation}
% \epsilon_2 \leq \frac{\epsilon \Delta}{6\mu^3} \leq  \frac{\epsilon \Delta}{6 n_h^3 \log^3 \left( \frac{6 M_1 \norm{a}_1 e^2 \lvert \lambda_{max}\rvert \pi  \norm{\trh{e^{-H}}}}{\epsilon} \right)}
%\end{equation}

\begin{align}
\epsilon_2 &\leq \frac{\epsilon}{15  M_1 \norm{a}_1 \pi \mu   \norm{\mathcal{L}'_{\mu,j}(\theta_j)}_{\infty} }\nonumber\\
	&\leq \frac{\epsilon \Delta}{15 M_1 \|a\|_1 \pi \left(n_h+\log(M_1 \Lambda/\epsilon)\right)^2 \log((n_h/2) + \log(M_1\Lambda/\epsilon)/2)}\nonumber \\
 	&\leq \frac{\epsilon \Delta}{15 M_1 \|a\|_1 \pi \mu^2 \log((\mu-1)/2)}
\end{align}
We hence obtain that the approximation error due to the divided differences and Fourier series approximation of $\sigma_v$ is bounded by $\epsilon/3$ for the above choice of $\epsilon_2$ and $\mu$. This bounds the second term in  \eqref{eq:bounds_ddfs} by $\epsilon/3$.\\

  Finally, we need to take into account the error $ \left\lVert \partial_{\theta} [\log_{K_1,M_1} \tilde{\sigma_v} - \log^s_{K_1,M_1} \tilde{\sigma}_v] \right\rVert$ which we introduce through the sampling process, i.e., through the finite sample estimate of $\mathbb{E}_s[\cdot]$ here indicated with the superscript $s$ over the logarithm.
Note that this error can be bound straight forward by \eqref{eq:full_approximation}.
We only need to bound the error introduced via the finite amount of samples we take, which is a well-known procedure.
The concrete bounds for the sample error when estimating the expectation value are stated in the following lemma.
\begin{lemma}
	Let $\sigma_{m}$ be the sample standard deviation of the random variable
	\begin{equation}
	\label{eq:random_variable_for_trace}
	\tilde{\mathbb{E}}_{s \in[0,1]} \left[\tr{ \rho e^{\frac{i s \pi m}{2}\sigma_v} e^{\frac{i \pi m'}{2} \sigma_v(\theta_j)} e^{\frac{i (1-s) \pi m}{2}\sigma_v} } \right],
	\end{equation}
	 such that the sample standard deviation is given by $\sigma_k = \frac{\sigma_m}{\sqrt{k}}$.
	Then with probability at least $1-\delta_s$, we can obtain an estimate which is within $\epsilon_s \sigma_m$ of the mean by taking $k=\frac{4}{\epsilon_s^2}$ samples for each sample estimate and taking the median of $O(\log(1/\delta_s))$ such samples.
\end{lemma}
\begin{proof}
From Chebyshev's inequality taking $k=\frac{4}{\epsilon_s^2}$ samples implies that with probability of at least $p=3/4$ each of the mean estimates is within $2\sigma_k = \epsilon_s \sigma_m$ from the true mean.
Therefore, using standard techniques, we take the median of $O( \log(1/\delta_s))$ such estimates which gives us with probability $1-\delta_s$ an estimate of the mean with error at most $\epsilon_s \sigma_m$, which implies that we need to repeat the procedure $O \left(\frac{1}{\epsilon_s^2}\log \left(\frac{1}{\delta_s}\right) \right)$ times.
\end{proof}
We can then bound the error of the sampling step in the final estimate, denoting with $\epsilon_s$ the sample error, as
\begin{align}
	\sum_{m=-M_1}^{M_1} \sum_{m'=-M_2}^{M_2}  \left\lvert \frac{i c_m c_{m'} m \pi}{2}  \right\rvert \sum_{j=0}^{\mu} \lvert \mathcal{L}'_{\mu,j}(\theta) \rvert \epsilon_s \sigma_m \nonumber\\
	\leq \frac{5 \norm{a}_1 M_1  \epsilon_s \sigma_m \pi \mu^2\log\left(\frac{\mu}{2}\right) }{\Delta} \leq \frac{\epsilon}{3},
\end{align}

We hence find that for
\begin{align}
	\epsilon_s &\leq \frac{\epsilon \Delta }{15 \norm{a}_1 M_1  \sigma_m \pi \mu^2\log\left(\frac{\mu}{2}\right) } \nonumber \\
	&\leq \frac{\epsilon \Delta}{15 M_1 \|a\|_1  \sigma_m \pi \left( n_h + \log(M_1 \Lambda/\epsilon)\right)^2 \log((n_n/2) + \log(M_1\Lambda/\epsilon)/2)}\nonumber \\
 	&\leq \frac{\epsilon \Delta}{15 M_1 \|a\|_1  \sigma_m \pi \mu^2 \log(\mu/2)}
\end{align}
also the last term in \eqref{eq:bounds_approximation_error} can be bounded by $\epsilon/3$, which together results in an overall error of $\epsilon$ for the various approximation steps, which concludes the proof.
\end{proof}
Notably all quantities which occure in our bounds are only polynomial in the number of the qubits.
The lower bounds for the choice of parameters are summarized in the following.
\begin{align}
M_1 &\geq \sqrt{L \log\left(\frac{9 \norm{\frac{\partial \sigma_v}{\partial \theta}} K_1 L \pi}{2\epsilon}\right)} \label{eq:bound_M1} \\
M_2 &\geq \left\lceil \log \left(\frac{4}{\epsilon_2} \right) \sqrt{(2\log{\delta_u^{-1}})^{-1}} \right\rceil \label{eq:bound_M2} \\
K_1 &\geq \log((1-\norm{\sigma_v})\epsilon/9)/\log(\norm{\sigma_v}) \label{eq:bound_K1} \\
K_2 &\geq \frac{\log(4/\epsilon_2)}{\log(\delta_u^{-1})} \label{eq:bound_K2} \\
L &\geq \frac{\log \left(\frac{\epsilon}{9\pi K_1 \norm{\frac{\partial \sigma_v}{\partial\theta}}}\right)}{\log(\norm{\sigma_v} \pi)} \label{eq:bound_L} \\
\mu &\geq n_h + \log \left( \frac{6 M_1 \norm{a}_1 e^2 \lvert \lambda_{max}\rvert \pi  \norm{\trh{e^{-H}}}}{\epsilon} \right):= n_h + \log(M_1 \Lambda/\epsilon) \label{eq:bound_mu} \\
\epsilon_2 &\le \frac{\epsilon \Delta}{15 M_1 \|a\|_1 \pi \left(n_h+\log(M_1 \Lambda/\epsilon)\right)^2 \log((n_h/2) + \log(M_1\Lambda/\epsilon)/2)}  \nonumber \\
	 &\leq \frac{\epsilon \Delta}{15 M_1 \|a\|_1 \pi \mu^2 \log((\mu-1)/2)} \label{eq:bounds_eps2}\\
\epsilon_s &\leq  \frac{\epsilon \Delta}{15 M_1 \|a\|_1  \sigma_m \pi \left( n_h + \log(M_1 \Lambda/\epsilon)\right)^2 \log((n_n/2) + \log(M_1\Lambda/\epsilon)/2)} \nonumber \\
		&\leq \frac{\epsilon \Delta}{15 M_1 \|a\|_1  \sigma_m \pi \mu^2 \log(\mu/2)} \label{eq:bounds_epss}
\end{align}

\subsection{Operationalising}
\label{app:proof_general_algo}
In the following we will make use of two established subroutines, namely sample based Hamiltonian simulation (aka the LMR protocol)~\cite{lloyd2014quantum}, as well as the Hadamard test, in order to evaluate the gradient approximation as defined in \eqref{eq:full_approximation}.
In order to hence derive the query complexity for this algorithm, we only need to multiply the cost of the number of factors we need to evaluate with the query complexity of these routines.
For this we will rely on the following result.
\begin{theorem}[Sample based Hamiltonian simulation~\cite{kimmel2017hamiltonian}]
  \label{thm:sample_based_ham_sim}
	Let $0 \leq \epsilon_h \leq 1/6$ be an error parameter and let $\rho$ be a density for which we can obtain multiple copies through queries to a oracle $O_{\rho}$.
	We can then simulate the time evolution $e^{-i\rho t}$ up to error $\epsilon_h$ in trace norm as long as $\epsilon_h/t \leq 1/(6 \pi)$ with $\Theta(t^2/\epsilon_h)$ copies of $\rho$ and hence queries to $O_{\rho}$.
\end{theorem}
We in particularly need to evaluate terms of the form
\begin{equation}
  \label{eq:error_formula_sample_based_ham_sim}
	\tr{ \rho e^{\frac{i s \pi m}{2}\sigma_v} e^{\frac{i \pi m'}{2} \sigma_v(\theta_j)} e^{\frac{i (1-s) \pi m}{2}\sigma_v}}
\end{equation}
Note that we can simulate every term in the trace (except $\rho$) via the sample based Hamiltonian simulation approach to error $\epsilon_h$ in trace norm.
This will introduce a additional error which we need to take into account for the analysis.
Let $\tilde{U}_{i}, i \in \{1,2,3\}$ be the unitaries such that $\norm{U_i - \tilde{U}_i}_* \leq \epsilon_h$ where the $U_i$ are corresponding to the factors in \eqref{eq:error_formula_sample_based_ham_sim}, i.e.,
$U_1 := e^{\frac{i s \pi m}{2}\sigma_v}$, $U_2 := e^{\frac{i \pi m'}{2} \sigma_v(\theta_j)}$, and $U_3 := e^{\frac{i (1-s) \pi m}{2}\sigma_v}$.
We can then bound the error as follows. First note that $\norm{\tilde U_i} \leq \norm{\tilde U_i - U_i} + \norm{U_i} \leq 1+\epsilon_h$, using Theorem~\ref{thm:sample_based_ham_sim} and the fact that the spectral norm is upper bounded by the trace norm.
\begin{align}
\label{eq:error_propagation_ham_sim}
  \tr{ \rho U_1 U_2 U_3 } &-	\tr{ \rho \tilde U_1 \tilde U_2 \tilde U_3 }  \leq \nonumber \\
  &= \tr{\rho U_1 U_2 U_3 -	\rho \tilde U_1 \tilde U_2 \tilde U_3} \nonumber \\
  &\leq \norm{U_1 U_2 U_3 -	\tilde U_1 \tilde U_2 \tilde U_3} \nonumber \\
  &\leq \norm{U_1 -	\tilde U_1} \norm{\tilde U_2} \norm{\tilde U_3}+ \norm{U_2 -	\tilde U_2} \tilde{U_3} + \norm{U_3 -	\tilde U_3} \nonumber \\
  &\leq \norm{U_1 -	\tilde U_1}_* (1+\epsilon_h)^2 + \norm{U_2 -	\tilde U_2}_* (1+\epsilon_h)+ \norm{U_3 -	\tilde U_3}_* \nonumber \\
  &\leq  \epsilon_h (1+\epsilon_h)^2 + \epsilon_h (1+\epsilon_h) + \epsilon_h = O(\epsilon_h),
\end{align}
neglecting higher orders of $\epsilon_h$, and where in the first step we applied the Von-Neumann trace inequality and the fact that $\rho$ is Hermitian, and in the last step we used the results of Theorem~\ref{thm:sample_based_ham_sim}.
We hence require $O((\max \{ M_1, M_2\} \pi)^2/\epsilon_h)$ queries to the oracles for $\sigma_v$ for the evaluation of each term in the multi sum in \eqref{eq:full_approximation}.
Note that the Hadamard test has a query cost of $O(1)$.
In order to hence achieve an overall error of $\epsilon$ in the gradient estimation we require the error introduced by the sample based Hamiltonian simulation also to be of $O(\epsilon)$.
In order to do so we require $\epsilon_h \leq O(\frac{\epsilon \Delta }{5 \norm{a}_1 M_1 \pi \mu^2 \log(\mu/2)})$, similar to the sample based error which yield the query complexity of
\begin{equation}
  \label{eq:final_comp_complexity_grad_entry_est}
  O\left( \frac{\max \{ M_1, M_2\}^2 \norm{a}_1 M_1  \pi^3 \mu^2 \log(\mu/2)}{\epsilon \Delta}\right)
\end{equation}
Adjusting the constants gives then the required bound of $\epsilon$ of the total error and the query complexity for the algorithm to the Gibbs state preparation procedure is consequentially given by the number of terms in \eqref{eq:full_approximation} times the query complexity for the individual term, yielding
\begin{align}
    O\left(\frac{M_1^2 M_2 \max \{ M_1, M_2\}^2 \norm{a}_1 \sigma_m \pi^3 \mu^3 \log \left(\frac{\mu}{2}\right)}{\epsilon\ \epsilon_s^2 \Delta} \right),
\end{align}
and classical precomputation polynomial in $M_1,M_2,K_1,L,s,\Delta,\mu$, where the different quantities are defines in eq.~(\ref{eq:bound_M1}-\ref{eq:bounds_epss}).

  Taking into account the query complexity of the individual steps then results in Theorem~\ref{thm:complexity_general_algo}.
  We proceed by proving this theorem next.
%   \begin{theorem}
%     \label{thm:complexity_general_algo}
%     Let $\rho, \sigma_v$ being two density matrices, $\norm{\sigma_v} <1/\pi$, and we have access to an oracle $O_{H}$ for the $d$-sparse Hamiltonian $H(\theta)$ and an oracle $O_{\rho}$ which returns copies of  purified density matrix of the data $\rho$, and $\epsilon \in (0,1/6)$ an error parameter.
%     With probability at least $2/3$ we can obtain an estimate $\mathcal{G}$ of the gradient w.r.t. $\theta \in \mathbb{R}^D$ of the relative entropy $\nabla_{\theta} \tr{\rho \log \sigma_v}$ such that
%     \begin{equation}
%       \norm{\nabla_{\theta}\tr{\rho \log \sigma_v} - \mathcal{G}}_{max} \leq \epsilon,
%     \end{equation}
%     with
%   \begin{align}
% \label{eq:final_query_complexity}
%      \tilde{O} \left(   \sqrt{\frac{N}{z}}
%       \frac{D \norm{H(\theta)}
%         d   \mu^5
%       \gamma
%       }
%       {\epsilon^3}
%     \right),
%   \end{align}
%   queries to $O_H$ and $O_{\rho}$, where $\mu\in O(n_h + \log(1/\epsilon))$, $\|\partial_{\theta} \sigma_v\| \le e^\gamma$, $\norm{\sigma_v}\geq 2^{-n_v}$ for $n_v$ being the number of visible units and $n_h$ being the number of hidden units, and $$\tilde{O}\left(\text{poly}\left(\gamma, n_v, n_h,\log(1/\epsilon)\right)\right)$$ classical precomputation.
%   \end{theorem}
  \begin{proof}[Proof of Theorem~\ref{thm:complexity_general_algo}]
    The runtime follows straight forward by using the bounds derived in \eqref{eq:final_comp_complexity_grad_entry_est} and Lemma~\ref{lem:proof_error_bound_grads}, and by using the bounds for the parameters $M_1,M_2, K_1,L,\mu,\Delta,s$ given in eq.~(\ref{eq:bound_M1}-\ref{eq:bounds_epss}).
For the success probability for estimating the whole gradient with dimensionality $d$, we can now again make use of the boosting scheme used in \eqref{eq:probability_boost} to be
\begin{align}
\label{eq:final_query_complexity_Gibbs_exact}
    \tilde{O} \left(
      \frac{
        d \norm{a}_1^3 \sigma_m^3 \mu^5 \log^3(\mu/2)
        \mathrm{polylog} \left( \frac{\norm{\frac{\partial \sigma_v}{\partial \theta}}}{\epsilon}, \,
         \frac{n_h^2 \norm{a}_1 \sigma_m}{\epsilon \Delta}\right)
      }
      {\epsilon^3 \Delta^3}
      \log \left( d \right)
    \right),
\end{align}
where $\mu=n_h + \log(M_1 \Lambda/\epsilon)$.\\
Next we need to take into account the errors from the Gibbs state preparation given in Lemma~\ref{thm:sample_based_ham_sim}.
For this note that the error between the perfect Hamiltonian simulation of $\sigma_v$ and the sample based Hamiltonian simulation with an erroneous density matrix denoted by $\tilde{U}$, i.e., including the error from the Gibbs state preparation procedure, is given by
\begin{align}
	\norm{\tilde U - e^{-i\sigma_v t}} &\leq \norm{\tilde U - e^{-i \tilde \sigma_v t}} +  \norm{e^{-i \tilde \sigma_v t} - e^{-i\sigma_v t}} \nonumber \\
	&\leq \epsilon_h  + \epsilon_G t
\end{align}
where $\epsilon_h$ is the error of the sample based Hamiltonian simulation, which holds since the trace norm is an upper bound for the spectral norm, and $\norm{\sigma_v - \tilde{\sigma}_v} \leq \epsilon_G$ is the error for the Gibbs state preparation from Theorem~\ref{thm:Gibbs_state_prep} for a $d$-sparse Hamiltonian, for a cost  $$\tilde{\mathcal{O}} \left(\sqrt{\frac{N}{z}} \norm{H}d \log \left( \frac{\norm{H}}{\epsilon_G}\right) \log \left( \frac{1}{\epsilon_G} \right) \right).$$
From \eqref{eq:error_propagation_ham_sim} we know that the error $\epsilon_h$ propagates nearly linear, and hence it suffices for us to take $\epsilon_G \leq \epsilon_h/t$ where $t = O(\max\{M_1,M_2\})$ and adjust the constants $\epsilon_h \leftarrow \epsilon_h/2$ in order to achieve the same precision $\epsilon$ in the final result.
We hence require
\begin{equation}
	\tilde{\mathcal{O}} \left( \sqrt{\frac{N}{z}} \norm{H(\theta)} \log \left( \frac{\norm{H(\theta)} \max \lbrace M_1,M_2\rbrace }{\epsilon_h}\right) \log \left( \frac{\max \lbrace M_1,M_2\rbrace }{\epsilon_h} \right) \right)
\end{equation}
and using the $\epsilon_h$ from before we hence find that this s bound by
\begin{equation}
	\tilde{\mathcal{O}} \left( \sqrt{\frac{N}{z}} \norm{H(\theta)} \log \left( \frac{\norm{H(\theta)} n_h^2}{\epsilon \Delta}\right) \log \left( \frac{n_h^2}{\epsilon \Delta} \right) \right)
\end{equation}
query complexity to the oracle of $H$ for the Gibbs state preparation. \\
The procedure succeeds with probability at least $1-\delta_s$ for a single repetition for each entry of the gradient.
In order to have a failure probability of the final algorithm of less than $1/3$, we need to repeat the procedure for all $D$ dimensions of the gradient and take for each the median over a number of samples.
Let $n_f$ be as previously the number of instances of the one component of the gradient such that the error is larger than $\epsilon_s \sigma_m$ and $n_s$ be the number of instances with an error $\leq \epsilon_s \sigma_m$ , and the result that we take is the median of the estimates, where we take $n=n_s+n_f$ samples.
The algorithm gives a wrong answer for each dimension if $n_s \leq \left\lfloor \frac{n}{2} \right\rfloor$, since then the median is a sample such that the error is larger than $\epsilon_s \sigma_m$.
Let $p=1-\delta_s$ be the success probability to draw a positive sample, as is the case of our algorithm.
Since each instance of (recall that each sample here consists of a number of samples itself) from the algorithm will independently return an estimate for the entry of the gradient, the total failure probability is bounded by the union bound, i.e.,
\begin{equation}
\pr_{fail} \leq D \cdot \pr \left[n_s \leq  \left\lfloor \frac{n}{2} \right\rfloor \right] \leq D \cdot e^{- \frac{n}{2(1-\delta_s)} \left((1-\delta_s) - \frac{1}{2} \right)^2} \leq \frac{1}{3},
\end{equation}
which follows from the Chernoff inequality for a binomial variable with $1-\delta_s>1/2$, which is given in our case for a proper choice of $\delta_s <1/2$.
Therefore, by taking $n \geq \frac{2-2\delta_s}{(1/2-\delta_s)^2} \log(3D) = O(\log(3D))$, we achieve a total failure probability of at least $1/3$ for a constant, fixed $\delta_s$.
Note that this hence results in an multiplicative factor of $O(\log(D))$ in the query complexity of \eqref{eq:final_query_complexity}.\\
The total query complexity to the oracle $O_{\rho}$ for a purified density matrix of the data $\rho$ and the Hamiltonian oracle $O_H$ is then given by
 \begin{equation}
     \tilde{O} \left(
       \sqrt{\frac{N}{z}}
	\frac{
        d  \log \left( d \right) \norm{H(\theta)}  \norm{a}_1^3 \sigma_m^3 \mu^5 \log^3(\mu/2)
        \mathrm{polylog} \left( \frac{\norm{\frac{\partial \sigma_v}{\partial \theta}}}{\epsilon}, \,
         \frac{n_h^2 \norm{a}_1 \sigma_m}{\epsilon \Delta}, \, \norm{H(\theta)} \right)
      }
      {\epsilon^3 \Delta^3}
    \right),
\end{equation}
which reduces to
  \begin{align}
    \tilde{O} \left(   \sqrt{\frac{N}{z}}
      \frac{D \norm{H(\theta)}
        d   \mu^5
      \alpha
      }
      {\epsilon^3}
    \right),
  \end{align}
hiding the logarithmic factors in the $\tilde{O}$ notation.
\end{proof}

\end{document}